\documentclass[11pt]{article}
\usepackage[hmargin=1in,vmargin=1in]{geometry}
\usepackage[dvipsnames,usenames,table]{xcolor}
\usepackage{hyperref}
\hypersetup{colorlinks=true, urlcolor=Blue, citecolor=Green, linkcolor=BrickRed, breaklinks, unicode}
\usepackage{amsmath,amssymb,amsfonts,amsthm}
\usepackage{graphicx}
\usepackage{bbm} 
\usepackage{xspace} 
\usepackage{bbm}
\usepackage{framed} 
\usepackage{thmtools}
\usepackage{thm-restate}
\usepackage{nicefrac} 
\usepackage[capitalise]{cleveref} 
\usepackage{mathtools}

\usepackage{setspace}
\newtheorem{theorem}{Theorem}[section]
\newtheorem{corollary}[theorem]{Corollary}
\newtheorem{lemma}[theorem]{Lemma}
\newtheorem{proposition}[theorem]{Proposition}
\newtheorem{definition}[theorem]{Definition}
\newtheorem{remark}[theorem]{Remark}
\newtheorem{hypothesis}[theorem]{Hypothesis}
\newtheorem{claim}[theorem]{Claim}

\newtheorem{fact}[theorem]{Fact}

\newtheorem{question}[theorem]{Question}

\newcommand{\R}{\mathbb R}

\newcommand{\eps}{\varepsilon}

\newcommand{\poly}{\mathrm{ poly}}
\newcommand{\cost}{\mathrm{cost}}
\newcommand{\polylog}{\mathrm{ polylog}}
\renewcommand{\tilde}{\widetilde}

\newcommand{\tG}{\tilde{G}}

\newcommand{\td}{\tilde{d}}

\newcommand{\tV}{\tilde{V}}

\newcommand{\E}{\mathbb{E}}

\newcommand{\calS}{\mathcal {S}}

\newcommand{\Oh}{{O}} 
\newcommand{\tOh}{\widetilde{\Oh}} 

\DeclareMathOperator{\dist}{dist}

\newcommand{\e}{\mathbf{e}}

\title{Near-Optimal Bounds for Parameterized Euclidean k-means}
 \author{Vincent Cohen-Addad\\ Google Research\\ \texttt{cohenaddad@google.com}\and Karthik C.\ S.\footnote{This work was supported by the National Science Foundation under Grants CCF-2313372 and CCF-2443697, a grant from the Simons Foundation, Grant Number 825876, Awardee Thu D. Nguyen, and partially funded by the Ministry of Education and Science of Bulgaria's support for INSAIT, Sofia University ``St. Kliment Ohridski'' as part of the Bulgarian National Roadmap for Research Infrastructure.}\\ Rutgers University\\\vspace{0.1in} \texttt{karthik.cs@rutgers.edu}\and David Saulpic \\ Université Paris Cité, CNRS\\ \texttt{david.saulpic@irif.fr}\and Chris Schwiegelshohn\footnote{This work was partially supported by the Independent Research Fund
 Denmark (DFF) under a Sapere Aude Research Leader grant No 1051-00106B and by a Google Research Award.}\\ Aarhus University\\ \texttt{schwiegelshohn@cs.au.dk} }
\date{}

\begin{document}

\maketitle
\begin{abstract}
The $k$-means problem is a classic objective for 
modeling clustering in a metric space. Given a set of points in a metric
space, the goal is to find $k$ representative points so as to minimize the sum of the squared distances from each point to its closest representative.
In this work, we study the approximability of $k$-means in Euclidean spaces parameterized by the number of clusters, $k$.\vspace{0.1cm}

In seminal works, de la Vega, Karpinski, Kenyon, and Rabani [STOC'03] and Kumar, Sabharwal, and Sen [JACM'10] showed how to obtain a $(1+\varepsilon)$-approximation for
high-dimensional Euclidean $k$-means in time $2^{(k/\varepsilon)^{O(1)}} \cdot dn^{O(1)}$.\vspace{0.1cm}

In this work, we introduce a new fine-grained hypothesis called \emph{Exponential Time for Expanders Hypothesis} (XXH) which roughly asserts that there are no non-trivial exponential time approximation algorithms for the vertex cover problem on near perfect vertex expanders. Assuming XXH, we close the above long line of work on approximating Euclidean $k$-means by showing that there is no $2^{(k/\varepsilon)^{1-o(1)}} \cdot n^{O(1)}$ time algorithm
achieving a $(1+\varepsilon)$-approximation for $k$-means in Euclidean space. 
This lower bound is tight as it matches the algorithm given by Feldman, Monemizadeh, and Sohler [SoCG'07] whose runtime is $2^{\tilde O (k/\varepsilon)} + O(ndk)$.  \vspace{0.1cm}

Furthermore, assuming XXH, we show that the seminal $\Oh(n^{kd+1})$ runtime exact algorithm of Inaba, Katoh, and Imai [SoCG'94] for $k$-means is optimal for small values of $k$.
\end{abstract}

\thispagestyle{empty}
\clearpage

\setcounter{page}{1}

\section{Introduction}

The $k$-clustering problem represents a fundamental task in data mining and machine learning, providing a model for grouping data points based on similarity. Given a set of points $P$ in a metric space $(X, \Delta)$, the objective is to select a set $C \subseteq X$ of $k$ points, referred to as \emph{centers}, and the goal is to minimize an objective function typically defined as the sum of the $z$-th powers of the distances from each point $p \in P$ to its nearest center in $C$. This general formulation encompasses several widely studied problems: $k$-median corresponds to $z=1$, $k$-means uses $z=2$ (minimizing the sum of squared distances), and $k$-center is when $z \rightarrow \infty$ (minimizing the maximum distance).

The algorithmic exploration of $k$-means, arguably the most popular variant, gained significant traction with the seminal work of Lloyd~\cite{lloyd1982least}. Since then, the problem has attracted substantial attention across diverse research communities, including operations research, machine learning, and theoretical computer science.

The computational complexity of $k$-means problem inherently depends on the structure of the underlying metric space $(X,d)$ and this paper focuses on the setting where points reside in Euclidean space $\mathbb{R}^d$. From a complexity theoretic perspective, Euclidean $k$-means is known to be NP-Hard even under seemingly restricted conditions, such as when the points lie in the Euclidean plane ($\mathbb{R}^2$) but $k$ is part of the input~\cite{megiddo1984complexity, MahajanNV12}, or when $k=2$ but the dimension $d$ is large~\cite{dasgupta2009random,AKP24}. On the algorithmic side, the best-known exact algorithm for Euclidean $k$-means was proposed by Inaba, Katoh, and Imai~\cite{InabaKI94}, running in time $\Oh(n^{kd+1})$. Remarkably, this has remained the state-of-the-art exact algorithm for over three decades.

Given the hardness of finding exact solutions, particularly in high dimensions, a significant line of research has focused on developing efficient \mbox{$(1+\eps)$-approximation} algorithms. Early breakthroughs by Fernandez de la Vega, Karpinski, Kenyon, and Rabani \cite{VegaKKR03} and Kumar, Sabharwal, and Sen \cite{KumarSS10} demonstrated that a $(1+\eps)$-approximation for high-dimensional Euclidean $k$-means could be achieved in time $2^{(k/\eps)^{\Oh(1)}} \cdot dn^{\Oh(1)}$. Subsequent improvements, leveraging techniques such as \textit{coresets} by Feldman, Monemizadeh, and Sohler \cite{FeldmanMS07}, refined the runtime to $\Oh(ndk) + 2^{\tOh(k/\eps)}$. Further work by Jaiswal, Kumar, and Sen \cite{Jaiswal0S14}, followed by Jaiswal, Kumar, and Yadav \cite{jaiswal2015improved}, utilized a simpler approach based on \emph{$D^2$-sampling} to achieve $\Oh(nd \cdot 2^{\tOh(k/\eps)})$ runtime (see also \cite{AbbasiBBCGKMSS23,bhore2025coresetstrikesbackimproved} for a generalized approach). 

All these different approaches seem to have hit a running-time barrier at $2^{k/\eps}$ (up to logarithmic factors in the exponent). This contrasts sharply with the related $k$-center problem, where Agarwal and Procopiuc \cite{agarwal2002exact} showed that a $(1+\eps)$-approximation can be obtained much faster, in time $n \log k + (k/\eps)^{\Oh(k^{1-1/d})}$. This runtime for $k$-center is known to be essentially optimal under the Exponential Time Hypothesis (ETH)~\cite{DBLP:conf/compgeom/ChitnisS22}. Furthermore, there exists a close relationship between $k$-means and the Partial Vertex Cover (PVC) problem; indeed, many known hard instances for $k$-means are derived from hard PVC instances \cite{AwasthiCKS15,LeeSW17,Cohen-AddadS19}. Interestingly, Manurangsi~\cite{manurangsi2018note} demonstrated that PVC admits a $(1+\eps)$-approximation in time $\eps^{-\Oh(k)} n^{\Oh(1)}$ (where $k$ is the solution size). This implies that the specific $k$-means instances derived from PVC can be approximated more efficiently than the runtime of the current best $(1+\eps)$-approximation algorithms  for  Euclidean $k$-means. 

This discrepancy motivates the central question of our work: When parameterized by the number of clusters $k$, is the current exponential dependency in $k/\varepsilon$ for approximating Euclidean $k$-means inherent, signifying a fundamental computational gap compared to $k$-center and PVC-related instances? Or, is it possible to devise a significantly faster approximation algorithm, potentially achieving a runtime closer to those known for related problems and offering a more unified algorithmic picture for $k$-clustering? More concretely:

\begin{question}\label{q:approx}
 Is it possible to design a $(1+\eps)$-approximation to $k$-means running in time $O(nd) + \eps^{-O(k)}$?
\end{question}

A negative answer to the above question would also imply progress toward understanding the hardness of exact algorithms. Currently, we are not aware of any fine-grained lower bound that matches the algorithm of Inaba et al.~\cite{InabaKI94}. The closest result we know of is by Cohen-Addad, de Mesmay, Rotenberg, and Roytman~\cite{CohenaddadSODA18}, who showed that if the centers must be selected from a prescribed set of ``candidate centers'', then no exact algorithm with a runtime of \(n^{o(k)}\) exists for \(k\)-median or \(k\)-means, even when the dimension is as low as four. It remains an open problem whether a similar result holds for the classic version where centers can be placed anywhere in \(\mathbb{R}^d\), and whether an \(n^{o(d)}\) lower bound also applies when \(k\) is constant.

A negative answer to \Cref{q:approx} would
answer those two questions: modern dimension-reduction \cite{MakarychevMR19} and coreset computation \cite{Cohen-AddadSS21} reduce the dimension and the number of distinct input points while preserving the clustering cost within a \((1 \pm \varepsilon)\) factor. Thus, as we show formally in \Cref{sec:exact}, any lower bound for approximation algorithms translates into a lower bound for exact algorithms, making progress toward demonstrating the optimality of the algorithm from~\cite{InabaKI94}.

Negative answers to such questions typically arise from \emph{fine-grained complexity} assumptions, such as the Exponential Time Hypothesis (ETH) \cite{IP01,IPZ01}, or the more modern Gap Exponential Time Hypothesis (Gap-ETH) \cite{MR16,D16}. Indeed these two assumptions have been very fruitful in explaining the intractability of various important geometric optimization problems \cite{Marx08,LokshtanovMS11,BergBKMZ20,Kisfaludi-BakNW21}. However, \Cref{q:approx} involves two parameters, \(k\) and \(\varepsilon\). Suppose we aim to rule out \(2^{(k/\varepsilon)^{1-o(1)}} \cdot \text{poly}(n, d)\)-time algorithms, then our parameter of interest is the quantity \(k/\varepsilon\), where \(k\) and \(\varepsilon\) are free variables constrained only by the fixed ratio \(k/\varepsilon\). To the best of our knowledge, such results are not known under ETH or Gap-ETH. In fact, proving such results entail several technical challenges which we discuss in \Cref{sec:introtechdiff}.


\subsection{Our Results}
In this paper we make substantial progress towards answering these questions.
 We introduce a hypothesis morally capturing a Gap Exponential Time Hypothesis for Vertex Cover on Near Perfect Vertex Expanders, and we call it \emph{E\textbf{\underline x}ponential Time for E\textbf{\underline x}panders \textbf{\underline H}ypothesis} (denoted {\bf XXH}). This hypothesis is formally stated in \cref{sec:hyp}, but an informal discussion about the statement of this hypothesis is given in \Cref{sec:introXXH}. Assuming XXH, we are able to answer \cref{q:approx} in the following way:

\begin{theorem}[Answer to \cref{q:approx}; informal statement]\label{thm:introlb}
Assuming XXH, for every $\beta>0$, there is no randomized algorithm running in time $2^{\left(k/\varepsilon\right)^{1-\beta}} \cdot \poly(n,d)$ that can $(1+\varepsilon)$-approximate the Euclidean $k$-means problem whenever $k\gg 1/\varepsilon$.

\end{theorem}
The formal statement of the lower bound is \cref{thm:lb_k}. 
\begin{sloppypar}This result might appear surprising due to the known close connection between vertex cover and \(k\)-means. Specifically, a partial vertex cover can be approximated within a \((1+\varepsilon)\) factor in \(\varepsilon^{-O(k)} n^{O(1)}\) time~\cite{manurangsi2018note} (where \(k\) is the solution size). This might lead one to expect an approximation scheme for \(k\)-means with similar complexity. 
\end{sloppypar}

\begin{sloppypar}For example, reductions such as those in Cohen-Addad et al.~\cite{CohenaddadSODA18} and Awasthi et al.~\cite{AwasthiCKS15} (see also~\cite{LeeSW17}) transform the vertex cover instance into a $k$-means instance by creating a point for each edge of the input graph and positioning the points in space such that pair of edges sharing a common vertex are close to each other. The hardness proof then relies on distinguishing between two kinds of instances:\end{sloppypar}

(1) Instances derived from graphs admitting a vertex cover of size \(k\), i.e., the edge set can be partitioned to $k$ stars. In this case, the corresponding point set can be partitioned into \(k\) clusters such that all points within the same cluster are close (say distance 1, representing edges covered by the same vertex from the cover).

(2) Instances derived from graphs where any set of \(k\) vertices leaves at least a constant fraction (say \(\delta\) fraction) of the edges uncovered. Consequently, in any partitioning of the corresponding point set into \(k\) clusters, either a constant fraction of these clusters contain points that are far apart (say distance 3, i.e., representing pairs of edges that do not share one of the selected \(k\) vertices) or a few clusters contain a lot of points (and most pairs are far apart).

The hardness for \(k\)-means then follows from the difficulty of distinguishing between these two cases based on the clustering cost. Specifically, it involves separating instances admitting a lower cost (e.g., \(n\), associated with Case 1) from those necessitating a higher cost (e.g., \((1-\delta)n + 2\delta n = n+\delta n\), associated with Case 2), where \(\delta\) is related to the minimum fraction of uncovered edges in the latter case.

However, since partial vertex cover can be approximated within 
a $(1+\eps)$ factor, for any $\eps>0$, in time $\eps^{-O(k)} n^{O(1)}$, and since 
  the objective scales linearly with the number
of edges not covered (i.e., clients that are at distance 3,
instead of 1, from their center), this type of instance can be solved in time $\eps^{-O(k)} n^{O(1)}$. Therefore, one cannot expect
to boost the lower bound running time from the partial vertex cover result to $2^{(k/\eps)^{1-o(1)}}$.

Thus, to show that the  $k$-means problem require an exponential dependency in
$\eps$, we need to develop a novel reduction framework.

To establish the above conditional lower bound for the Euclidean \(k\)-means problem, we first reduce the vertex cover problem in the non-parameterized setting to an intermediate graph problem, essentially in the parameterized setting, but with reduced structure so as to fail standard algorithmic techniques for partial \(k\)-vertex cover, while retaining enough structure to embed the graph problem into the Euclidean \(k\)-means problem. We direct the reader to \cref{sec:overview} for further details where we also try to clarify how vertex expansion in the vertex cover problem helps us overcome several technical difficulties.

\Cref{thm:introlb} implies lower-bounds for exact algorithms as well:
\begin{corollary}
\label{cor:exact}
Assuming XXH, for every $\beta>0$, there is no algorithm that for any $k,d$
solves the Euclidean $k$-means problem in time $n^{(k\sqrt{d})^{1-\beta}}$, nor in time $n^{d^{1-\beta}}$.
\end{corollary}
Note that this result contrasts the running time of $k$-means clustering with that of $k$-center, which admits an exact algorithm running in time $O\left(n^{k^{1-1/d}}\right)$ \cite{agarwal2002exact}. The current state of the art algorithm enumerating over all Voronoi partitions in time $O(n^{kd+1})$ by Inaba, Katoh and Imai~\cite{InabaKI94} is thus a likely candidate for being optimal -- in particular, it is almost optimal for constant $k$.

\subsection{Further Related Work}
\paragraph{Hardness of Clustering.}
As we mentioned previously, the $k$-means and $k$-median problems are NP-hard, even when $k=2$ (and $d$ is large) \cite{dasgupta2009random,AKP24}, or when $d=2$ (and $k$ is large) \cite{megiddo1984complexity, MahajanNV12}. When both parameters are part of the input, the problems become APX-hard \cite{GuK99, JMS02, GuI03, AwasthiCKS15, Cohen-AddadSL21, Cohen-AddadSL22}.
Most techniques to show hardness of approximation are based on reducing from covering problems to clustering problems, for instance through structured instances of max $k$-coverage or set cover.
Recent works have used different approaches: \cite{Cohen-AddadSL21} showed how to use hardness of some coloring problems to prove hardness of approximation for $k$-median and $k$-means in general metric spaces,
and \cite{Cohen-AddadSL22} focused on Euclidean spaces and tried to pinpoint what combinatorial structures allow for gap-preserving embeddings to Euclidean space.

For general metrics, the connection between $k$-clustering and the set cover problem (or rather max-coverage) has been known since the fundamental work of Guha and Khuller~\cite{GuK99}, who established
the best known hardness of approximation bounds.
This connection was observed again when analyzing the parameterized complexity of the problem: Cohen-Addad, Gupta, Kumar, Lee and Li~\cite{Cohen-AddadG0LL19} showed how to approximate $k$-median and $k$-means up to factors
$1+2/e$ and $1+8/e$ respectively and showed that this is tight assuming Gap-ETH (see also~\cite{Cohen-AddadL19,manurangsi2020tight,DBLP:conf/esa/AdamczykBMM019,DBLP:conf/cocoon/ChenHXXZ23,DBLP:journals/jco/HanXDZ22}).

\paragraph{FPT algorithms via Sketching.} The past decades have seen the development of very powerful sketching and compression methods that allow reducing
the dimension to $O(\log k / \varepsilon^{-2})$ \cite{MakarychevMR19} and the number of distinct input points to $\tilde{O}(k \varepsilon^{-z-2})$ ($z=1$ for $k$-median, $z=2$ for $k$-means) via the construction of coresets \cite{FL11, BravermanJKW21, huang2020coresets, Cohen-AddadSS21, stoc22}. Perhaps surprisingly, these bounds are independent of the original input size and dimension and can be computed in near-linear time, which allows for the construction of simple FPT algorithms.
Applying the $O(n^{kd + 1})$ algorithm of \cite{InabaKI94} indeed gives a complexity of $2^{\tilde{O}(k/\varepsilon^2)}$ plus the near-linear time to sketch the input; naively enumerating all partitions yields a running time of $k^{\tilde{O}(k/\varepsilon^{z+2})}$, plus the time to sketch the input.

We crucially remark that the dependency on $\varepsilon$ cannot be substantially improved: \cite{stoc22} showed a lower bound of $\Omega(k\varepsilon^{-2})$ for coresets, and \cite{LarsenN17}  showed the optimality of the dimension reduction.
Therefore, one cannot hope to go below $2^{\Omega(k/\varepsilon^2)}$ and answer \cref{q:approx} using only these techniques.

Other parameters were studied for $k$-means clustering: most notably, the cost has been investigated by Fomin, Golovach and Simonov \cite{FominGS21}. They presented a $D^D \mathrm{poly}(nd)$ exact algorithm for $k$-median, where $D$ is an upper bound on the cost.

\paragraph{From Continuous to Discrete Clustering.}\begin{sloppypar} The other standard technique to design FPT algorithms is to find a small set of candidate centers that contains a near-optimal solution. This approach was used, for instance, by \cite{BadoiuHI02} to obtain the first algorithm running in time \mbox{$2^{(k/\varepsilon)^{O(1)}} d^{O(1)} n \log^{O(k)}(n)$}, by \cite{KumarSS10} to improve the running time, and by \cite{BhattacharyaJK18} for the capacitated clustering problem.
\end{sloppypar}

\paragraph{Approximation Algorithms in Euclidean Spaces.} To compute a $(1+\varepsilon)$-approximation in time polynomial in $n$ and $k$, any algorithm must have a running time at least doubly exponential in $d$, as the problem is APX-hard in dimension $\Omega(\log n)$. The best of these algorithms is from \cite{jacm}, with a near-linear running time of $f(\varepsilon, d) n \operatorname{polylog} n$.
If one sticks to algorithms polynomial in $n, k$ and $d$, the lower bounds on the approximation ratio are $1.06$ for $k$-median and $1.015$ for $k$-means, conditioned on P~$\neq$~NP \cite{Cohen-AddadSL22}. The upper bounds are still quite far: $2.41$ for $k$-median, and $5.96$ for $k$-means \cite{euclideanApprox}.

\subsection{Organization of the paper}
We provide the proof overview of Theorem~\ref{thm:introlb} in Section~\ref{sec:overview}, and then in Section~\ref{sec:prelim} we provide some notations, definition, and tools useful for this paper.
Then, we present in \cref{sec:lb_k} our formal proof of Theorem~\ref{thm:introlb} based on the hypothesis XXH defined in \cref{sec:hyp}. Sections~\ref{sec:completeness} and \ref{sec:soundness} contain the completeness and soundness analysis of reduction underlying Theorem~\ref{thm:introlb} respectively. Finally, in Section~\ref{sec:exact} we prove Corollary~\ref{cor:exact}.

\section{Our Techniques}\label{sec:overview}

We would like to now convey the conceptual and technical ideas that went into proving the lower bound in Theorem~\ref{thm:introlb}. As alluded to earlier in this section, hard instances of Euclidean $k$-means are typically constructed from the Vertex Cover problem, where every edge is mapped to a client and the partition of the edge set (which is the client set) by an optimal vertex cover also yields the optimal clustering for the $k$-means objective.

\subsection{Motivation and Technical Background}\label{sec:introtechdiff}
\paragraph{Current Understanding of the Landscape.}
Starting from Gap-ETH (for 3-SAT)~\cite{D16,MR16}, it is easy to show that there is no $2^{o(n)}$ time algorithm to $1+\delta$ approximate the Vertex Cover problem (on sparse graphs) for some small constant $\delta>0$. By a standard reduction, this implies that Euclidean $k$-means cannot be approximated within $1+\delta$ in time $2^{o(n)}$, albeit when $k$ is linear in $n$ (the number of clients).
On the other hand, starting from ETH~\cite{IP01,IPZ01}, it is easy to show (for example following the reduction in \cite{dasgupta2009random} or \cite{AKP24}) that Euclidean $2$-means cannot be exactly solved in $2^{o(n)}$ time. 
Thus, we can show that there is no $2^{o(k/\varepsilon)}\cdot \poly(n)$ time algorithm for $1+\varepsilon$ approximating Euclidean $k$-means problem, when either (i) $k=\Omega(n)$ and $\varepsilon=\Omega(1)$, or (ii) $k=2$ and $\varepsilon=1/\Omega(n)$. 

\paragraph{Exploring Uncharted Territories.} From the above discussion, we know that if $k$ was $\Omega(n)$ then we cannot obtain a $2^{o(n)}$ time approximation algorithm, but what if $k$ was approximately $\sqrt{n}$? Then algorithmic techniques based on coresets yield a $(1+\varepsilon)$ approximation in time $2^{\tilde{O}(\sqrt{n}/\varepsilon)}$. But is it possible to do better?

We can try to answer this from the lower bound viewpoint. We can look at the above mentioned Gap-ETH hard instance, i.e., the setting $k=\Omega(n)$ and $\varepsilon=\Omega(1)$ and duplicate each point  $n$ times to obtain a point-set with $N=n^2$ points, and we have that there is no $2^{\sqrt{N}}$ time algorithm that can constant approximate the objective. But we cannot even rule out the possibility that there is an \emph{exact} algorithm running in time  $2^{N^{0.51}}$ time for this value of $k$. Ideally, we would like to be able to rule out algorithms running in time  $2^{N^{0.5+\rho}}$ that provide a $(1+\Omega(1/N^{\rho}))$ approximate solution, for every $\rho\in [0,0.5]$.
Thus, the result that we are shooting for is:
\begin{center}\emph{Rule out $2^{o(k/\varepsilon)}\cdot \poly(n)$ time algorithms, when $k/\varepsilon$ is fixed,}\end{center} 
i.e., we want the lower bound to hold on the entire  tradeoff curve between the number of clusters $k$ and the accuracy of clustering $\varepsilon$.
To the best of our knowledge, there are no such results known in fine-grained complexity. Thus, we are motivated to develop a new framework to prove such results. 

\paragraph{Technical  Challenges.}
One approach is to start from the lower bound given in $(ii)$, i.e., when $k=2$ and $\varepsilon=1/\Omega(n)$, and reduce it to a different instance of $k$-means, where $k$ has increased (say to $\sqrt{n}$), but also $\varepsilon$ has increased to $1/\sqrt{n}$ (all this with a linear blowup in size). But this requires ``gap creation'', a notorious challenging task, potentially much harder than even proving Gap-ETH from ETH! 

Therefore, we pursue the approach of  starting from the lower bound given in $(i)$, i.e., when $k=\Omega(n)$ and $\varepsilon=\Omega(1)$, and reduce it to a different instance of $k$-means, where $k$ has decreased (say to $\sqrt{n}$), and $\varepsilon$ has also decreased to $1/\sqrt{n}$. This is the approach of trading off the gap to reduce the number of clusters. 

\paragraph{A (Failed) Simple Approach.} Let $G=([n],E)$ be a vertex cover instance (where $|E|:=m=O(n)$) which is hard to approximate to $1+\delta$ factor (for some positive constant $\delta$) under Gap-ETH. Thus, size of an optimal vertex cover of $G$, denoted $\alpha_G$ is $\Omega(n)$. Suppose, our target $k$-means instance is when $k=o(m)$ and $\varepsilon=\Omega(\delta k/\alpha_G)$: then we may simply look at the embedding where each edge $\{u,v\}\in E$ is mapped to the point (i.e., client), $\mathbf{e}_u+\mathbf{e}_v\in \mathbb{R}^n$, where $\mathbf{e}_i$ is the standard basis vector which is 1 on $i^{\text{th}}$ coordinate, and 0 everywhere else. If we were asked to cluster this client set to $\alpha_G$ clusters minimizing the $k$-means objective in Euclidean metric, then the optimal solution would simply be the partition based on some optimal vertex cover of $G$. Thus, we could hope that even when asked to cluster the clients to $k$ parts, the optimal solution would be to first cluster the clients into $\alpha_G$ clusters (based on the vertex cover) and then merge clusters, so as to end up with only $k$ clusters in the end.  However, since $k\ll \alpha_G$, we can have near-optimal cost from clustering which do not correspond to any vertex cover of $G$. For example, a typical obstacle  is when we have the following clustering: $k-1$ clusters each contain a single client, and one cluster contains all the remaining points. It is entirely possible that such clusterings also have low cost.

\paragraph{Embedding via Color Coding.}
To overcome the above issue (of imbalancedness) of clustering, we introduce a color coding based embedding technique. Given $G$ and a target number of clusters $k$, we first uniformly and independently randomly color each vertex $v$ in $G$, with a color in $[k]$ (and let color of $v$ be denoted by $c_v$). Thus each edge (consisting of two vertices), also gets (at most) two colors. Now, consider
the embedding where each edge $\{u,v\}\in E$ is mapped to the point (i.e., client), $\mathbf{e}_u+\mathbf{e}_v+\mathbf{e}_{c_u}+\mathbf{e}_{c_v}\in \mathbb{R}^{n+k}$. Consider the clustering of this client set  to $k$ parts in the following way. Let $S\subset [n]$ be any optimal vertex cover of $G$ and $S:=S_1\dot\cup \cdots \dot\cup S_k$ be a partition of the vertices in the vertex cover based on the coloring. Then the alleged optimal clusters would be given by the edges covered by $S_1,\ldots ,S_k$, i.e., cluster $i$ would be all edges (i.e., the corresponding points of these edges) in $G$ covered by the vertices in $S_i$ (break ties arbitrarily). 

The embedding of the colors of the edges, forces the optimal clusters to have a dominating color, and since the colors are uniformly spread, the obstacle mentioned in the above approach (without coloring), would yield high clustering cost and thus can now be ruled out. 

That said, there is more serious obstacle  that is not addressed: when we merge clusters in the completeness and soundness case, the graph topology  affects the $k$-means cost; we elaborate on this next.

\paragraph{Need for Vertex Expansion.} Consider two graphs $G_1$ and $G_2$, both on $n$ vertices and are $d$-regular. Suppose $G_1$ looks like a random $d$ regular graph, and thus is a very good vertex expander, and for every subset $S$ of vertices of size $O(n/d)$ we have that the number of unique neighbors is about $|S|\cdot (d-O(1))$ \cite{Vadhan12}. On the other hand suppose that $G_2$ is obtained by first taking $n/(d-3)$ disjoint copies of $K_{d-3,d-3}$ and then adding a random $3$-regular graph to connect these copies. 

We encounter the following problem when we merge clusters as described above: $G_1$ behaves in the way we expect it to, whereas the $G_2$ has a very low cost even if it doesn't have a good vertex cover. This is because if we take the $d-3$ sized independent set in a $K_{d-3,d-3}$ then we do get a very good vertex cover of $K_{d-3,d-3}$ but in addition if all (or even a large fraction of) the edges of $K_{d-3,d-3}$ were in the same cluster then each edge has $2d-2$ other edges that is adjacent to it in the cluster. On the other hand, if we take the edges of an independent set in $G_1$, then a typical edge has $d-1$ other edges incident to it. This makes the analysis of completeness and soundness impossible without knowing more about the graph topology.

\subsection{New Hypothesis: Exponential Time for Expanders Hypothesis}\label{sec:introXXH}

\paragraph{Exponential Time for Expanders Hypothesis (XXH).} To remedy the situation we introduce a working hypothesis that the gap vertex cover problem cannot be solved in time $2^{n^{1-o(1)}}$ on random graphs. For the sake of keeping the proofs clean (to the extent possible), we state the hypothesis in Section~\ref{sec:hyp} in terms of vertex expanders which makes it directly usable.  Under this hypothesis, we use the color coding embedding that was described earlier and with a lot of technical effort are able to derive the conclusion given in Theorem~\ref{thm:introlb}.

Informally, XXH asserts that for some constants $\zeta,\delta$, such that $\zeta\ll \delta$,
no randomized algorithm running in $2^{n^{1-o(1)}}$ time can take as input  a  $d$-regular vertex expander $G:=([n],E)$ and 
 distinguish between the following:
\begin{description}
\item[Completeness:] There exist $n/2$ vertices that covers at least $(1-\zeta)$ fraction of $E$. 
\item[Soundness:]  Every subset of $V$ of size $n/2$ does not cover $\delta$ fraction of the edges.
\end{description}

\paragraph{Plausibility of XXH.}
A key observation connecting random graphs and expander graphs is that random $d$-regular graphs are known to be very strong vertex expanders with constant probability (see Theorem 4.2 in \cite{Vadhan12}). Consequently, making progress on understanding the inapproximability of the Vertex Cover problem on vertex expanders is closely related to the hardness of the vertex cover problem on random graphs.

XXH formalizes this hardness on vertex expanders and, via our main reduction (Theorem~\ref{thm:lb_k}), links it directly to the Euclidean $k$-means problem. This connection implies we live in one of three possible worlds. The first is where XXH is false and the Euclidean $k$-means admits a much faster algorithm (such as $2^{(k/\varepsilon)^{0.99}}\cdot \poly(n)$ time), and this would yield new algorithms for computing Vertex Cover on random-like graphs. The second is where XXH is false but there are no new faster algorithms for Euclidean $k$-means, and this would force a deeper, computational understanding of vertex expanders themselves. The third world is where XXH is true, and thus the current Euclidean $k$-means algorithm is nearly optimal. The reader is directed to Section~\ref{sec:plausibilityXXH} for more discussion on these three possibilities.

Investigating the truth of XXH (the third world) leads to natural questions of independent interest. For instance, Ramanujan graphs are known to be near-perfect vertex expanders for small sets~\cite{mckenzie_et_al}, so we ask: does the Vertex Cover problem admit a PTAS  when the input graph is a Ramanujan graph? Progress on this question would directly help us better understand XXH.

As partial evidence for XXH, in Corollary~\ref{cor:AKS} we show that if we forego the expansion property, then under the Unique Games Conjecture it is possible to show from \cite{AKS11} that no polynomial-time algorithm can take as input a $d$-regular graph $G:=([n],E)$ and distinguish between the following:
\begin{description}
\item[Completeness:] There exist $n/2$  vertices that cover at least a $(1-\zeta)$ fraction of $E$. 
\item[Soundness:] Every subset of $V$ of size $n/2$ fails to cover an $\Omega(\sqrt{\zeta})$ fraction of the edges.
\end{description}

\paragraph{Our Message.}
We wish to highlight two key aspects of the XXH hypothesis. First, it fundamentally acts as a fine-grained complexity assumption regarding the hardness of the Vertex Cover problem on random instances. Second, it establishes an important connection: developing better approximation algorithms for Euclidean $k$-means would provide non-trivial approaches to solving Vertex Cover on expander graphs, a problem of significant independent interest. In addition, even if weaker versions of XXH are true (and proved in the future), this would imply weaker lower bounds for  $(1+\varepsilon)$-approximating the Euclidean $k$-means problem (see Remark~\ref{rem:weakXXH} for details).

\subsection{Proof Overview of Theorem~\ref{thm:introlb}}

We now give an overview of the completeness and soundness cases. Recall that we are given a $d$-regular graph $G=([n],E)$, and we have constructed a point $\mathbf{e}_u+\mathbf{e}_v+\mathbf{e}_{c_u}+\mathbf{e}_{c_v}\in \{0,1\}^{n+k}$, for each edge $(u,v)$ whose end points have colors $c_u$ and $c_v$. 

In the completeness case, we have that there are $n/2$ vertices cover $1-\zeta$ fraction of vertices, so our strategy to cluster is straightforward: first form $n/2$ clusters, each one corresponding to a vertex in the vertex cover solution, and then identify each cluster with a color by looking at the color of the common vertex in each cluster (which is a star graph). Then, each color class would be a cluster, and we have $k$ clusters, and we can show that the cost is $3|E|-(1-7\zeta)kd$.  

However, our soundness analysis is highly non-trivial, and involves several tools and arguments of different flavors. Suppose we have a clustering of $P$ whose $k$-means cost is about $3|E|-(1-\delta^5)kd$ (where $\zeta\le \delta^5/10$), then  we first connect the cost of cluster $C_i$ (for $i\in [k]$) with certain properties of the graph $G_i$ in the following way:
\begin{lemma}[Informal version of Lemma~\ref{lem:costCluster}]\label{lem:informalcost}
 For every cluster $C_i$ we have its cost is equal to: \begin{align*}
 3|C_i|-1+ (\gamma_i - \kappa_i)  (|C_i|-1)-\frac{1}{|C_i|}\cdot \sum_{v \in V_i}  d_{i,v}^2,
\end{align*}  
where $d_{i,v}$ is the degree of $v$ in the graph $G_i$ (induced by edges in $C_i$), $\gamma_i$ be the fraction of pairs of edges in $G_i$ that have no color in common, and $\kappa_i$ be the fraction of pairs of edges that have two colors in common.
\end{lemma}

Next, we show that if $\gamma_i$ (fraction of pairs of edges with no color in common) and $\kappa_i$ (fraction of pairs of edges with two colors in common) are not too large, then there is a dominant single color in the cluster $C_i$. 
\begin{lemma}[Informal version of Lemma~\ref{lem:monochrome}]
If $\kappa_i$ and $\gamma_i$ are bounded by some small constants, then there is a large  fraction of edges in $G_i$ that have the same color. 
\end{lemma}

Then, we show that we can identify a large subcollection of clusters for which both $\kappa_i$ and $\gamma_i$ are small. In addition, we assume that the cost of the clustering $C_1,\ldots ,C_k$ is at most $3|E| - (1-\delta^5)dk$. Then, for each cluster in this subcollection,  we can relate a bound on the sum of the squared degrees of the vertices in the subcollection to the clustering cost appearing in Lemma~\ref{lem:informalcost}.

\begin{lemma}[Informal version of Lemma~\ref{lem:avgStruct}] 
 There is some $I\subseteq [k]$ such that  for all $i\in I$, it holds that $\kappa_i $ and $\gamma_i $ are small and  $\sum_{i\in I}|E_i|\ge (1-38\delta^2)\cdot |E|$. 
In addition, we also have for all $i\in I$:
 \begin{align*}
  \frac{\sum_{v\in G_i}d_{i,v}^2 }{\sum_{v\in G_i}d_{i,v}} \geq (1-\delta^3) 2d.
\label{eqgood}
\end{align*}
\end{lemma}

Finally, we show that a cost of clustering $3|E|-(1-\delta^5)dk$ implies that we can construct a set of vertices $S\subseteq [n]$ of size slightly more than $n/2$ such that we cover $(1-10\delta^{1.5})$ fraction of  the edges, contradicting the soundness assumption of XXH for small enough $\delta$. This step is quite involved, and we skip providing more details about it here.

\section{Preliminaries}\label{sec:prelim}

\paragraph{Notations}

We consider the Euclidean space $\R^d$, with the $\ell_2$ metric: $\dist(p, q) := \sqrt{\sum_{i=1}^d (p_i - q_i)^2}$.
The $k$-means cost function of $P$ using set of centers $\calS$ is defined as $\cost(P, \calS) := \sum_{p \in P} \dist(p, \calS)^2$.

\begin{definition}[Continuous $k$-means]
 Given a set $P \subset \R^d$, the \emph{continuous} $k$-means problem is the task of finding $k$ centers $\calS$ in $\R^d$ that minimize the cost function $\cost(P, \calS)$.
\end{definition}

The continuous $k$-means problem has a surprising (and folklore) formula that expresses the cost of a clustering only in terms of pairwise distances between input points:
\begin{lemma}\label{lem:folklore}[Folklore]
 For a given cluster $C \subset \R^d$, the optimal center is the average $\frac{\sum_{p \in C} p}{|C|}$, and the $k$-means cost is
 $\sum_{p_1, p_2 \in C} \frac{\|p_1-p_2\|^2}{2|C|}$.
\end{lemma}

Sometimes the points of $P$ are called \emph{clients}. A \emph{solution} is any subset of $\R^d$ with size $k$.
We say a solution is an $\alpha$-approximation if its cost is at most $\alpha$ times the minimal cost.

We use the following two Chernoff-type bounds, where the latter is suited to partially dependent variables:

\begin{theorem}[Chernoff bounds]\label{thm:chernoff}
Suppose $X_1, \dots, X_n$ are i.i.d. random Bernoulli variables with parameter $p$, and let $\mu = np$ be the expected value of their sum. Then the following bounds hold:
\begin{description}
    \item[Upper tail:] For any $0 < \eps < 1-p$,
\[
\Pr\!\left[\sum_{i=1}^n X_i \ge (p+\eps)n\right]
\le \exp\!\bigl(-n D(p+\eps \| p)\bigr),
\]
where
\[
D(q\|p)= q\ln\frac{q}{p} + (1-q)\ln\frac{1-q}{1-p}.
\]
    \item[Lower tail:] For any $0 < \lambda < 1$,
\[
\Pr\!\left[\sum_{i=1}^n X_i \le (1-\lambda)\mu\right]
\le \exp\!\left(-\frac{\lambda^2\mu}{2}\right).
\]
\end{description}
\end{theorem}

\section{Exponential Time for Expanders Hypothesis}\label{sec:hyp}

In this section, we formally introduce the Exponential Time for Expanders Hypothesis which is then used in the next section to prove conditional lower bounds for the Euclidean $k$-means problems for small $k$.

\subsection{New Hypothesis: Exponential Time for Expanders Hypothesis}

We first define the notion of vertex expansion relevant to this submission.

\begin{definition}[Small Set Vertex Expanders]\label{def:smallSetExpander}
 Given constants $\alpha>0$, a $d$-regular graph $G=(V,E)$ on $n$ vertices, and an integer $k:=k(n)$, we say that $G$ is a $(k,\alpha)$-small set vertex expander if for every subset $S\subseteq V$ of size at most $n/k$ we have that $|\{u\in V: \exists \{u,v\}\in E \text{ and }v\in S\}|\ge (1-\alpha)\cdot d\cdot |S|$.
\end{definition}

Now, we can define our new hypothesis.

\begin{hypothesis}[Exponential Time for Expanders Hypothesis -- XXH$(\delta, \zeta, \alpha)$]\label{hypo:xxh}
Given constants $\delta, \zeta, \alpha \in (0, 1)$, the XXH$(\delta, \zeta, \alpha)$ assumption states that the following holds for all $\beta>0$:
No randomized algorithm running in $2^{n^{1-\beta}}$ time can, given as input a $d$-regular $(\polylog n, \alpha)$-small set vertex expander $G=(V=[n],E)$ with $d=(\log n)^{L}$ (for some $L>1$), distinguish between the following with probability 0.9:
\begin{description}
\item[Completeness:] There exist $n/2$ vertices that cover at least $(1-\zeta)$ fraction of $E$.
\item[Soundness:] For every $S\subseteq V$ of size $n/2$, there are at least $\delta\cdot |E|$ many edges which are not covered by any vertex in $S$.
\end{description}
\end{hypothesis}

\subsection{Plausibility of XXH: Three Possible Worlds}\label{sec:plausibilityXXH}
For the sake of the discussion in this subsection, we refer to the event ``Clustering Barrier is Breached'' to simply denote the existence of an algorithm much faster than $2^{k/\varepsilon}\cdot \poly (n)$ for $1+\varepsilon$ approximating the Euclidean $k$-means problem. In Section~\ref{sec:lb_k}, we proved that assuming XXH, the clustering barrier cannot be breached. Therefore, we are living in one of three possible worlds. The first world is where XXH is false and the clustering barrier is breached. The second world is where XXH is false, but the clustering barrier cannot be breached. Finally, the third world is where XXH is true (and thus the clustering barrier cannot be breached from Theorem~\ref{thm:lb_k}).

\emph{The message we want to convey here is that  only one of the above three worlds is possible, and regardless of which world is proven to be the one we live in, it will shed new light on problems of interest to the algorithmic community.}

\subsubsection{World I: XXH is False and Clustering Barrier is Breached}
XXH may be viewed as a fine-grained assumption for the vertex cover problem on random instances, which we elaborate on below. Thus, our reduction from XXH instances to Euclidean $k$-means instances, as given in Theorem~\ref{thm:lb_k}, can be used to make oracle calls to the efficient algorithm for the Euclidean $k$-means problem (since the clustering barrier is breached in this world) to solve XXH instances efficiently (i.e., in mildly sub-exponential time). This would imply a separation in this world between worst-case and average-case instances of the gap vertex cover problem (brushing aside many details to make a succinct claim). 

\paragraph{XXH as a Fine-Grained Assumption for Random Instances.}
Theorem 4.2 in \cite{Vadhan12} shows that for some large universal constant $C$, a random $d$-regular graph is a $(C,\frac{2}{d})$-small set vertex expander with a probability of 0.5. Thus, making progress on the inapproximability of the vertex cover problem for vertex expanders is morally similar to proving the hardness of approximation for vertex cover on random $d$-regular graphs. Or in this world, we would make algorithmic progress on understanding vertex cover on random graphs through (hypothetical) clustering algorithms.

\subsubsection{World II: XXH is False but Clustering Barrier is not Breached}
In this world, XXH is false, possibly because\footnote{We are not addressing here the concern that XXH might be false in this world because of the setting of the parameters.} the vertex cover problem is computationally easy on small set vertex expanders. For \emph{spectral} expanders, it is sometimes possible to apply Hoffman’s ratio bound \cite{haemers2021hoffman} to obtain non-trivial speedups when the spectral gap is large. However, vertex expanders, while intuitively similar to edge expanders (and spectral ones), are poorly understood. In fact, explicit constructions with parameters close to those obtained for random graphs were developed only recently \cite{vertexexpanders}. Thus, falsifying XXH motivates a better understanding of vertex expanders from the computational viewpoint of optimization problems.

Nevertheless, Theorem 11 in \cite{AbboudW23} shows that the \emph{exact} vertex cover problem remains as hard on spectral expanders as it is on general graphs.

\subsubsection{World III: XXH is True}
The first evidence that XXH might be true is that the hard instances of vertex cover constructed by \cite{DinurS02} are essentially built from random label cover instances which have strong expansion properties after PCP composition. Moreover, XXH only promises expansion for small sets (sub-polynomial size sets), and thus it is unlikely that such a local structure can be algorithmically used.

Also, it can be formally argued that on small sets, Ramanujan graphs are near-perfect vertex expanders \cite{mckenzie_et_al}. Therefore, as a way to prove XXH, one can first ask whether the vertex cover problem admits a PTAS on Ramanujan graphs. It is possible that the answer to this question is negative, although there are no techniques to handle such questions, and thus XXH opens this new line of exploration.

\subsection{Small Progress on XXH under Unique Games Conjecture}

In this subsection, we show that without the vertex expansion property, it is possible to obtain a weak version of XXH under the unique games conjecture \cite{khot2002power}. 

Let $\Phi$ denote the cumulative density function of the standard normal
distribution and, for any $\rho \in [-1, 1], \mu \in [0, 1]$, let $\Gamma_{\rho}(\mu)$ denote $\Pr[X \leq \Phi^{-1}(\mu) \wedge Y \leq \Phi^{-1}(\mu)]$ where $X, Y$ are normal random variables with means 0, variances 1 and covariance $\rho$. The main intermediate result of~\cite{AKS11} is the following:

\begin{theorem}[Theorem 1 from \cite{AKS11}]\label{thm:AKS}
For any $q \in (0, 1/2)$ and any $\varepsilon > 0$, it is UG-hard to, given a regular graph $G = (V, E)$, distinguish between the following two cases.
\begin{itemize}
\item (Completeness) $G$ contains an independent set of size at least $q \cdot |V|$.
\item (Soundness) For any subset $T \subseteq V$, the number of edges with both endpoint in $T$ is at least $|E| \cdot \left(\Gamma_{-q/(1 - q)}(\mu) - \varepsilon\right)$ where $\mu = |T|/|V|$.
\end{itemize}
\end{theorem}

This means that, in the completeness case, there is a vertex cover of size at most $(1 - q) \cdot |V|$. On the other hand, in the soundness case, if we consider any subset $S \subseteq V$ of size at most $(1 - q) \cdot |V|$, then the number of edges \emph{not} covered is exactly the same as the number of edges with both endpoints in $(V \setminus S)$, which is at least $(\Gamma_{-q/(1 - q)}(q) - \varepsilon) \cdot |E|$. We will evaluate this approximation factor for $q=\frac{1}{2}-\zeta$. 

\begin{claim}\label{claim:zeta}
 Let  $q:=\frac{1}{2}-\zeta$ where $\zeta>0$. Then for small enough $\zeta$, we have: $\Gamma_{-q/(1 - q)}(q) \ge \frac{\sqrt{\zeta}}{3}$.
\end{claim}

Assuming the above claim (which will be proved later), we have the following corollary from Theorem~\ref{thm:AKS} by setting $\varepsilon\ll \sqrt{\zeta}$:

\begin{corollary}\label{cor:AKS}
There exists $\delta,\zeta > 0$ with $\delta^2>\zeta/33$ such that it is UG-hard to, given a regular graph $G = (V, E)$, distinguish between the following two cases.
\begin{description}
\item[Completeness:] There exists $|V|/2$  vertices that covers at least $(1-\zeta)$ fraction of $E$. 
\item[Soundness:]  For every $S\subseteq V$ of size $|V|/2$, there are at least $\delta\cdot |E|$ many edges which are not covered by any vertex in $S$.
\end{description}
\end{corollary}
\begin{proof}
 Let $G=(V,E)$ be a $d$-regular hard instance given by Theorem~\ref{thm:AKS} plugging in $q:=\frac{1}{2}-\zeta_0$. Then, from Claim~\ref{claim:zeta}, we have that either $G$ has a vertex cover of size at most $\left(\frac{1}{2}+\zeta_0\right) \cdot |V|$, or for every subset $S \subseteq V$ of size at most $\left(\frac{1}{2}+\zeta_0\right) \cdot |V|$,  the number of edges \emph{not} covered by $S$ is  at least $(\frac{\sqrt{\zeta_0}}{3} - \varepsilon) \cdot |E|$. Moreover, in the completeness case, by an averaging argument, there is a subset $S^* \subseteq V$ of size $|V|/2$ such that $S^*$ covers at least $\frac{\nicefrac{1}{2}}{\nicefrac{1}{2}+\zeta_0}$ fraction of edges. Note that $\frac{\nicefrac{1}{2}}{\nicefrac{1}{2}+\zeta_0}>1-2\zeta_0$. On the other hand, in the soundness case, by setting $\varepsilon=o(\sqrt{\zeta_0})$, we can conclude that for every subset $S \subseteq V$ of size  $ |V|/2$,  the number of edges \emph{not} covered by $S$ is  at least $\frac{\sqrt{\zeta_0}}{4}  \cdot |E|$. Then, we define $\delta:= \frac{\sqrt{\zeta_0}}{4}$ and $\zeta:= 2\zeta_0$ to obtain the theorem statement. Note that $\delta^2=\frac{\zeta}{32}> \frac{\zeta}{33}$.
\end{proof}
 
This corollary is weaker than the promise given in XXH in the following ways. First, and most importantly, the hard instances of Corollary~\ref{cor:AKS} need not be vertex expanders. Second, the conditional lower bound under the Unique Games Conjecture is only against polynomial-time algorithms, whereas XXH rules out sub-exponential time algorithms.  

Between the two remarks made above, the first one is the major obstacle, and this is mainly to do with our state of poor understanding of vertex expanders. Once the  toolkit develops on this topic, it is conceivable that some additional progress can be made on XXH.

We close this section, with the proof of Claim~\ref{claim:zeta}.
\begin{proof}[Proof of Claim~\ref{claim:zeta}]Let $a := \Phi^{-1}(1/2 - \zeta)$. As $\zeta \to 0$, we have the expansions:$$ a \sim -\sqrt{2\pi}\zeta \quad \text{and} \quad \delta := 1+\rho = \frac{4\zeta}{1+2\zeta} \sim 4\zeta. $$Since $a < 0$, we have $\Gamma_{-1}(q) = \Pr[X \le a \wedge -X \le a] = 0$. We compute $\Gamma_{\rho}(q)$ by integrating the bivariate Gaussian density $\phi_2(a, a; r)$ from $r=-1$ to $\rho$. Letting $u = 1+r$, the integral is dominated by the behavior near $u=0$:\begin{align*}\Gamma_{\rho}(q) &= \int_{-1}^{\rho} \phi_2(a, a; r) ,dr\sim \int_{0}^{4\zeta} \frac{1}{2\pi\sqrt{2u}} \exp\left(-\frac{a^2}{u}\right) ,du.\end{align*}Substituting $v = a^2/u$ and using the asymptotic $\int_{x}^{\infty} v^{-3/2}e^{-v} dv \sim 2x^{-1/2}$ for small $x$:\begin{align*}\Gamma_{\rho}(q) &\sim \frac{|a|}{2\pi\sqrt{2}} \int_{\frac{a^2}{4\zeta}}^\infty v^{-3/2} e^{-v} ,dv\sim \frac{|a|}{2\pi\sqrt{2}} \cdot 2\left(\frac{|a|}{\sqrt{4\zeta}}\right)^{-1}= \frac{\sqrt{2\zeta}}{\pi}.\end{align*}Finally, since $\frac{\sqrt{2}}{\pi} \approx 0.45 > \frac{1}{3}$, the claim holds.\end{proof}

\section{Conditional Lower Bound for Euclidean $k$-means with Few Clusters}\label{sec:lb_k}
We prove in this section our main  result, the hardness of approximating $k$-means when parameterized by $k$. The formal theorem is:

\begin{theorem}[Fine-Grained Hardness for Approximation of Euclidean $k$-means from XXH]\label{thm:lb_k}
Suppose XXH$(\delta, \zeta, \alpha)$ is true for some constants $\delta, \zeta, \alpha \in (0, 1)$ satisfying $\delta \le 10^{-3}$, $\zeta \le \delta^5/10$, and $\alpha \le \delta^{10}$. Let $L>1$ be the constant defining the degree $d = (\log |V|)^L$ in the XXH hypothesis. 

Let $\tilde{k}:\mathbb{N}\to\mathbb{N}$ be a non-decreasing function and $\tilde \varepsilon:\mathbb{N}\to (0,1)$ be a non-increasing function. 
Define $f(n) := \frac{\tilde k(n)}{\tilde\varepsilon(n)}$. Assume that for all sufficiently large $n$, $f(n)$ satisfies the smoothness condition $f(n) \le C \cdot f(n-1)$ for some constant $C \ge 1$. Furthermore, assume the following asymptotic limits hold:
\begin{itemize}
    \item $\tilde k(1) \ge \frac{1}{\tilde{\varepsilon}(1)} \ge 2/\delta^5$.
    \item $f(n) = \omega(\log n)$ \quad and \quad $f(n) = o\left(\frac{n}{\log^L n}\right)$.
    \item $\tilde \varepsilon(n) = o\left(\frac{1}{(\log n)^{\omega(1)}}\right)$.
    \item $\tilde k(n) \cdot \tilde\varepsilon(n) = \omega\left((\log n)^{2L}\right)$.
\end{itemize}

Then the following holds for all $\beta>0$:
No randomized algorithm can, given as input exactly $n$ points in $\mathbb{R}^{\poly(n)}$ and an integer $\tilde k(n)$, run in $2^{f(n)^{1-\beta}}\cdot \poly(n)$ time and output a $(1+\tilde\varepsilon(n))$-approximate estimate of the $\tilde{k}(n)$-means cost with probability at least $0.9$.
\end{theorem}

In this section, we actually prove the following theorem, which applies for $\tilde {k}$ and $\tilde\varepsilon$ bounded in a specific way and then show that it immediately implies Theorem~\ref{thm:lb_k} above.

\begin{theorem}\label{thm:lb_k_n}
Let $\delta, \zeta, \alpha \in (0, 1)$ be some constants satisfying $\delta \le 10^{-3}$, $\zeta \le \delta^5/10$, and $\alpha \le \delta^{10}$.
There is a randomized algorithm running in linear time in the input size, which takes as input an integer $k$ (where $k>\sqrt{|V|}\cdot d$, $k=o(|V|/d^{\omega(1)})$) and a $d$-regular $(\polylog |V|, \alpha)$-small set vertex expander $G=(V,E)$ (where $d=(\log |V|)^{L}$, for some $L>1$), and outputs a point-set $P\subseteq \mathbb{R}^{O(|V|)}$ of at most $|E|$ points such that with probability at least $0.95$, the following holds:
   \begin{description}
       \item[Completeness:] If there are $|V|/2$ vertices covering at least $1-\zeta$ fraction of $E$, then there is a clustering of $P$ such that the $k$-means cost is at most $3|E|-(1-7\zeta)kd$.
       \item[Soundness:] If every $|V|/2$ vertices miss at least $\delta$ fraction of $E$, then every clustering of $P$ has $k$-means cost at least $3|E|-(1-\delta^5)kd$.
   \end{description}
\end{theorem}

Now the proof of Theorem~\ref{thm:lb_k} follows as an immediate corollary by a simple duplication technique.

\begin{proof}[Proof of Theorem~\ref{thm:lb_k}]
Assume, for the sake of contradiction, that there exists a constant $\beta \in (0,1)$ and a randomized algorithm $\mathcal{A}$ that, given $n$ points and a target cluster count $\tilde{k}(n)$, runs in time $2^{f(n)^{1-\beta}} \cdot \poly(n)$ and outputs a $(1+\tilde\varepsilon(n))$-approximation to the $\tilde k(n)$-means cost with probability at least $0.95$. We use $\mathcal{A}$ to refute the XXH hypothesis in sub-exponential time.

Let $G = (V, E)$ be an instance of XXH on $N = |V|$ vertices with degree $d = (\log N)^L$.  
Define $C_{\mathrm{gap}} := \frac{3}{\delta^5 - 7\zeta}$. The condition $\zeta \le \delta^5/10$ ensures $\delta^5 - 7\zeta \ge 0.3\delta^5 > 0$, so $C_{\mathrm{gap}}$ is well-defined and strictly positive. 
Since $f(n) = \omega(\log n)$ diverges, we let $n^*$ be the minimal integer such that $f(n^*) \ge C_{\mathrm{gap}} N$. By minimality, $f(n^*-1) < C_{\mathrm{gap}} N$. Applying the smoothness condition $f(n) \le C \cdot f(n-1)$, we obtain $C_{\mathrm{gap}} N \le f(n^*) < C \cdot C_{\mathrm{gap}} N$, establishing $f(n^*) = \Theta(N)$.

Let the target number of core clusters be $k_{\mathrm{core}} := \tilde k(n^*) - 1$. We verify that $k_{\mathrm{core}}$ satisfies the preconditions of Theorem~\ref{thm:lb_k_n} for $G$. Since $f(n^*) = \Theta(N)$ and $f(n) = o(n/\log^L n)$, we  have $\log N = \Theta(\log f(n^*)) =O(\log n^*)$.
\begin{itemize}
    \item \textbf{Lower bound $k_{\mathrm{core}} > \sqrt{N} d$:} We have $\frac{k_{\mathrm{core}}^2}{N} = \Theta\Big(\frac{\tilde k(n^*)^2}{f(n^*)}\Big) = \Theta\big(\tilde k(n^*) \tilde \varepsilon(n^*)\big)$. By the theorem's assumptions, this is $\omega\big((\log n^*)^{2L}\big) = \omega\big((\log N)^{2L}\big) = \omega(d^2)$. For sufficiently large $N$, this   implies $k_{\mathrm{core}} > \sqrt{N} d$.
    \item \textbf{Upper bound $k_{\mathrm{core}} = o\big(N / d^{\omega(1)}\big)$:} We have $\frac{k_{\mathrm{core}}}{N} \le \frac{\tilde k(n^*)}{N} = \Theta\big(\tilde \varepsilon(n^*)\big) = o\big((\log n^*)^{-\omega(1)}\big) = o\big((\log N)^{-\omega(1)}\big) = o\big(d^{-\omega(1)}\big)$.
\end{itemize}

We apply Theorem~\ref{thm:lb_k_n} to $(G, k_{\mathrm{core}})$, producing a point set $P_{\mathrm{core}} \subset \mathbb{R}^{O(N)}$ of size $M \le |E| = Nd/2$. Since $N = \Theta(f(n^*))$ and $\log N = O(\log n^*)$, we have $M = O(f(n^*) (\log n^*)^L) = o(n^*)$. Thus, for sufficiently large $N$, $M \ll n^*$.

Because $\mathcal{A}$ requires exactly $n^*$ points, we construct a padded instance $P_{\mathrm{new}}$ as follows.  Let $W := \lfloor (n^*-1) / M \rfloor$. Since $M = o(n^*)$, $W \ge 1$. We insert exactly $W$ identical copies of every point in $P_{\mathrm{core}}$ into $P_{\mathrm{new}}$, yielding $WM$ points. The remaining $R := n^* - WM \ge 1$ points  satisfy $R \le M$. We introduce a new orthogonal basis vector $\mathbf{e}_{\mathrm{new}}$ and add $R$ identical dummy points at coordinate $x^* := \Lambda \mathbf{e}_{\mathrm{new}}$. 

We choose $\Lambda > 0$ to be sufficiently large (e.g., computable in polynomial time as $\Lambda^2 > n^* \cdot W \sum_{x \in P_{\mathrm{core}}} \|x\|^2$). Because $x^*$ lies on an orthogonal axis, this choice guarantees that any clustering placing a dummy point and a core point in the same cluster incurs a variance penalty strictly greater than the baseline cost of clustering all core points at the origin and placing a dedicated center at $x^*$. Thus, any optimal clustering  dedicates exactly $1$ center to $x^*$ (incurring zero cost for the identical dummy points) and uses the remaining $k_{\mathrm{core}}$ centers to optimally partition the $W$ copies of $P_{\mathrm{core}}$.\footnote{If $\mathcal{A}$ requires distinct points, an infinitesimally small perturbation avoids multisets without meaningfully affecting the continuous objective gap.}

The total number of points is $WM + R = n^*$. We set $K := \tilde k(n^*) = k_{\mathrm{core}} + 1$ and run $\mathcal{A}$ on $(P_{\mathrm{new}}, K)$. 
Since exact point duplication scales the $k$-means variance objective uniformly, the optimal $K$-means cost is exactly $W \cdot \cost_{k_{\mathrm{core}}}(P_{\mathrm{core}})$. By Theorem~\ref{thm:lb_k_n}, we have:
\begin{description}
    \item[Completeness:] $\cost(P_{\mathrm{new}}) \le W \cdot C_{\mathrm{comp}} \le W \big(3|E| - (1-7\zeta)k_{\mathrm{core}}d\big)$.
    \item[Soundness:] $\cost(P_{\mathrm{new}}) \ge W \cdot C_{\mathrm{sound}} \ge W \big(3|E| - (1-\delta^5)k_{\mathrm{core}}d\big)$.
\end{description}

To successfully distinguish the two cases, the relative error of $\mathcal{A}$ must not bridge the gap: $(1+\tilde\varepsilon(n^*)) W C_{\mathrm{comp}} < W C_{\mathrm{sound}}$, which rearranges to:
\[ 
    \tilde\varepsilon(n^*) C_{\mathrm{comp}} < C_{\mathrm{sound}} - C_{\mathrm{comp}} = k_{\mathrm{core}} d (\delta^5 - 7\zeta). 
\]
Using the worst-case bound $C_{\mathrm{comp}} \le 3|E| = 1.5 N d$, it suffices to guarantee $1.5 N \tilde\varepsilon(n^*) < k_{\mathrm{core}} (\delta^5 - 7\zeta)$, or equivalently $N < \frac{\delta^5 - 7\zeta}{1.5} \frac{k_{\mathrm{core}}}{\tilde\varepsilon(n^*)}$. 
By our initial choice of $n^*$, we have $N \le \frac{\delta^5 - 7\zeta}{3} f(n^*) = \frac{\delta^5 - 7\zeta}{3} \frac{\tilde k(n^*)}{\tilde\varepsilon(n^*)}$.
Since $\tilde k(1) \ge 2/\delta^5 \ge 2000$ and $\tilde k$ is a non-decreasing function, $k_{\mathrm{core}} = \tilde k(n^*) - 1 \ge \frac{1}{2}\tilde k(n^*)$. This yields:
\[ N \le \frac{\delta^5 - 7\zeta}{3} \frac{k_{\mathrm{core}} + 1}{\tilde\varepsilon(n^*)} < \frac{\delta^5 - 7\zeta}{1.5} \frac{k_{\mathrm{core}}}{\tilde\varepsilon(n^*)} \]
Thus, $\mathcal{A}$ correctly distinguishes the instances.

By the union bound, the reduction mapping and $\mathcal{A}$ succeed simultaneously with probability at least $1 - 0.05 - 0.05 = 0.90$.
The total time to resolve the XXH instance comprises the reduction overhead and the execution of $\mathcal{A}$. 
Because $\tilde{\varepsilon}(n) = o\left(1 / (\log n)^{\omega(1)}\right)$, we have $f(n) = \omega((\log n)^c)$ for any constant $c>0$. This implies $\log n^* = N^{o(1)}$, yielding $n^* = 2^{N^{o(1)}}$. Thus, the reduction time to compute bounds and construct $P_{\mathrm{new}}$ is $\poly(N, n^*) = 2^{o(N)}$. 
Algorithm $\mathcal{A}$ runs in $2^{f(n^*)^{1-\beta}} \cdot \poly(n^*) = 2^{O(N^{1-\beta})} \cdot 2^{N^{o(1)}}$ time. 
For any constant $\gamma \in (0, \beta)$ and sufficiently large $N$, the overall runtime is  bounded by $2^{N^{1-\gamma}}$. This decides the XXH hypothesis in strictly sub-exponential time, producing a contradiction.
\end{proof}

\begin{remark}\label{rem:weakXXH}
 It is worth noting that the reduction in Theorem~\ref{thm:lb_k_n} applies even for weaker versions of XXH. For example, if one day in the future, XXH was proved against algorithms running in time $2^{\sqrt{n}}$ instead of the currently stated $2^{n^{1-o(1)}}$ runtime algorithms, then we can recover that there is no randomized algorithm running in time much faster than $2^{\sqrt{k/\varepsilon}} \cdot \poly(n,d)$ that can $(1+\varepsilon)$-approximate the Euclidean $k$-means problem. 
\end{remark}

The rest of this section is dedicated to proving Theorem~\ref{thm:lb_k_n}. In Section~\ref{sec:construction_k} we provide the reduction from the instance of the vertex cover problem given by XXH to the Euclidean $k$-means problems with the parameters given in Theorem~\ref{thm:lb_k_n}. Next in Sections~\ref{sec:completeness}~and~\ref{sec:soundness}, we prove the completeness and soundness properties of the reduction, thus completing the proof of Theorem~\ref{thm:lb_k_n}. Finally, in Section~\ref{sec:exact}, we show how Theorem~\ref{thm:lb_k} implies Corollary~\ref{cor:exact}.

\subsection{Construction of the instance}\label{sec:construction_k}

Our starting point is a $d$-regular $(\polylog n,\alpha)$-small set vertex expander $G=(V=[n],E)$ (where $\alpha \le \delta^{10}$, $k>\sqrt{n}\cdot d$ and $d=(\log n)^{L}$, for some $L>1$).

Let $\rho:[n]\to [k]$ be a uniform random function (also viewed as a coloring), where $\rho(j) = i$ with probability $1/k$ independently for all $j\in [n]$ and $i\in [k]$. Let $[n]=U_1\dot\cup\cdots \dot\cup U_{k}$ be the partition induced by $\rho$, i.e., for all $i\in[k]$, we have $U_i:=\{j\in [n]|\rho(j)=i\}$. 
Since $\rho$ is a uniform random function, applying Chernoff bound, we have:

\begin{proposition}\label{prop:balanced}
 With probability at least 0.99, for all $i\in [k]$, we have $|U_i| = \frac{n}{k} \pm 3\sqrt{\frac{n}{k}\cdot \log k}$.
\end{proposition}
\begin{proof}
 Fix an index $i \in [k]$. For each vertex $j \in [n]$, let $X_{j,i}$ be the indicator random variable for the event that vertex $j$ is assigned to set $U_i$.
 By the definition of $\rho$, the variables $\{X_{j,i}\}_{j=1}^n$ are independent Bernoulli trials with parameter $p = 1/k$. The expected size of the set $U_i$ is:
 \[
  \E[|U_i|] = \sum_{j=1}^n \E[X_{j,i}]   = \frac{n}{k}.
 \]
 
 We now apply the Chernoff bound (Theorem~\ref{thm:chernoff}) with the deviation term set to $3\sqrt{\frac{n}{k} \log k}$ to obtain the following:
 \[
  \Pr\left( \left| |U_i| - \frac{n}{k} \right| \ge 3\sqrt{\frac{n}{k} \log k} \right) \le 2 \exp\left( -\frac{9 \frac{n}{k} \log k}{3 \frac{n}{k}} \right)   = 2k^{-3}\le \frac{1}{100k},
 \]
when $k\ge 15$.

 By the Union Bound, the probability that the bound holds simultaneously for all $i \in [k]$ is at least $1-\frac{k}{100k}=0.99$.
\end{proof}

We next note that $\rho$ induces a coloring on the vertices of $G$, simply by labeling vertices $1$ to $k$.
For every edge $e\in E$, let $\rho(e)\subset [k]$ be the colors of the two endpoints of $e$, formalized as follows: if $e$ connects $u$ and $v$,
$$\rho(e) := \{\rho(u), \rho(v)\}.$$

\subsubsection{Properties of the random coloring}

Call an edge $e$ bad if $|\rho(e)|=1$. A simple probabilistic analysis bounds the number of bad edges:

\begin{proposition}\label{proposition:badEdges}
 With probability at least $0.99$, the fraction of bad edges is at most $100/k$.
\end{proposition}
\begin{proof}
Let $Y$ be the random variable counting the number of bad edges. We write $Y = \sum_{e \in E} I_e$, where $I_e$ is the indicator that edge $e=(u,v)$ is bad (i.e., $\rho(u)=\rho(v)$).
 
 Since vertex colors are chosen uniformly and independently, for any edge $e=(u,v)$, $\Pr(I_e=1) = \sum_{c=1}^k \Pr(\rho(u)=c \land \rho(v)=c) = \sum_{c=1}^k \frac{1}{k^2} = \frac{1}{k}$.
 
 By linearity of expectation, the expected fraction of bad edges is $\E[Y/|E|] = \frac{1}{|E|} \sum_{e \in E} \E[I_e] = \frac{1}{k}$. Applying Markov's inequality to the non-negative random variable $Y/|E|$:
 \[
  \Pr\left(\frac{Y}{|E|} \ge \frac{100}{k}\right) \le \frac{\E[Y/|E|]}{100/k} =   0.01.\qedhere
 \]
\end{proof}

Next, we call a pair of edges non-representative if $|\rho(e)\cup \rho(e')|<|e\cup e'|$. On the other hand, we call a pair of edges representative if $|\rho(e)\cup \rho(e')|=|e\cup e'|$. Similar to \Cref{proposition:badEdges}, we can bound the number of non-representative pairs:

\begin{proposition}\label{proposition:badpairs}
 With probability $0.99$, the fraction of non-representative pairs is at most $600/k$.
\end{proposition}
\begin{proof}
 Let $Z$ be the random variable counting the number of non-representative pairs. For any pair of edges $\{e, e'\}$, let $S = e \cup e'$ be the set of vertices involved. Since $|S| \le 4$, the probability of a collision in the coloring of $S$ is bounded by $\binom{4}{2} \frac{1}{k} = \frac{6}{k}$.

 By linearity of expectation, the expected fraction of non-representative pairs is at most $6/k$. Applying Markov's inequality:
 \[
  \Pr\left(\text{fraction of non-representative pairs} \ge \frac{600}{k}\right) \le \frac{6/k}{600/k} =   0.01.\qedhere
 \]
\end{proof}

A particular subset of non-representative pairs is those that have one vertex in common -- and hence have two colors in common. We bound more precisely those in the next lemma.

\begin{proposition}\label{proposition:badVertex}
 With probability $0.99$, the number of edges that have one vertex and two colors in common with another edge is at most $100 \frac{nd^2}{k}$.
\end{proposition}
\begin{proof}
 Let $W$ be the random variable counting the number of such edges. Such edges come in pairs sharing a vertex. For any vertex $u$, there are $\binom{d}{2}$ pairs of incident edges. Let $e=\{u,v\}$ and $e'=\{u,w\}$ be such a pair. They share two colors if $\rho(e)=\rho(e')$, which implies $\rho(v)=\rho(w)$ (occurring with probability $1/k$).
 
 Summing over all $n$ vertices, the expected number of such pairs is $n \binom{d}{2} \frac{1}{k} \le \frac{nd^2}{2k}$. Since each pair contributes at most 2 edges to the count $W$, we have $\E[W] \le 2 \cdot \frac{nd^2}{2k} = \frac{nd^2}{k}$. Applying Markov's inequality:
 \[
  \Pr\left(W \ge 100 \frac{nd^2}{k}\right) \le \frac{\E[W]}{100 nd^2/k} \le \frac{nd^2/k}{100 nd^2/k} = 0.01.\qedhere
 \]
\end{proof}

Finally, we bound the number of non-representative pairs in any set $C$. For this, we   use the Chernoff bounds given in Theorem~\ref{thm:chernoff}.

\begin{lemma}\label{lem:kappa_i}
With probability at least 0.99, it holds that for every set $C$ of edges, the fraction of pairs of edges in $C$ having two colors in common is at most $\frac{36 \log k}{|C|-1}$.
\end{lemma}
\begin{proof}
 For every pair of distinct colors $a,b\in [k]$ (with $a \neq b$), let $E_{ab}$ be the set of edges in the graph with colors $a$ and $b$. Note that monochromatic edges (where $a=b$) are considered ``bad'' edges and are excluded from the point set $P$ by definition in Section~\ref{sec:construction_k}. Thus $E_{aa} \cap C = \emptyset$ for any cluster $C$. The number of pairs of edges in $C$ having two colors in common is therefore at most $N_{ab}:=\sum_{a \neq b} \binom{|E_{ab} \cap C|}{2}$. Note that the total number of pairs of edges in $C$ is $\binom{|C|}{2} \le \frac{|C|^2}{2}$.

 The quantity $N_{ab}$ is maximized when the mass is concentrated on as few sets $E_{ab} \cap C$ as possible. Indeed, by convexity, if $x \geq y$, then for any $\eta > 0$, $x^2 + y^2 \leq (x+\eta)^2 + (y-\eta)^2$.

 Hence, we focus on a fixed pair of colors, and show that $|E_{ab}| \leq 36 \log k$ with high probability.
 
 First, from Proposition~\ref{prop:balanced}, the number of vertices with color $a$ is at most $2n/k$ with probability $1-\exp(-n/k)$.
 
 We show an additional property: any vertex has at most $4$ neighbors with color $a$, with probability $1-1/k^2$. To see this, consider the probability that a vertex has at least five neighbors with color $a$. This corresponds to a Binomial distribution with parameters $d$ and $1/k$; hence the probability is at most $\binom{d}{5}(1/k)^5 \le \frac{d^5}{k^5}$. The expected number of vertices with at least five neighbors of color $a$ is therefore at most $\frac{nd^5}{k^5}$. Markov's inequality ensures that, with probability $1-1/k^2$, this number is at most $\frac{nd^5}{k^3}$. Since $k \ge \sqrt{n}d$, we have $k^3 \ge n^{1.5}d^3 > nd^5$ (for sufficiently large $n$), implying this quantity is strictly less than 1, so no such vertex exists.

 Recall that $U_a$ is the set of vertices with color $a$. Since the graph is $d$-regular and $|U_a| \le 2n/k$, the set $U_a$ has at most $2nd/k$ outgoing edges. For a vertex $v \in U_a$, each neighbor receives color $b$ independently with probability at most $1/(k-1)$.

 Applying Chernoff bounds (Theorem~\ref{thm:chernoff}) with number of trials $N = 2nd/k$ and probability $p = 1/k$, the number of neighbors of $U_a$ that have color $b$ is bounded. Let this count be $S$. We bound $\Pr(S > Np + \beta N)$. We choose $\beta = \frac{4k \log k}{nd}$.

 Using the assumption $k \ge \sqrt{n}d$, we have $k^2 \ge nd^2$. Thus:
 \[
 \beta k = \frac{4k^2 \log k}{nd} \ge \frac{4nd^2 \log k}{nd} = 4d \log k \ge 8,
 \]
 (since $d \ge 2$ and $k$ is large). This implies $\beta \ge 8/k = 8p$, and consequently $p+\beta \le \frac{9}{8}\beta$.

 Using the bound $D(p+\beta, p) \ge \frac{\beta^2}{2(p+\beta)}$, the exponent is:
 \[
 \frac{\beta^2 N}{2(p+\beta)} > \frac{\beta^2 N}{2(\frac{9}{8}\beta)} = \frac{4}{9}\beta N = \frac{4}{9} \left( \frac{4k \log k}{nd} \cdot \frac{2nd}{k} \right) = \frac{32}{9} \log k > 3.5 \log k.
 \]
 Thus, $\Pr(S > Np + \beta N) \le k^{-3.5}$. The deviation term is $\beta N = 8 \log k$. The mean is $Np = \frac{2nd}{k^2} < 1$ (for large $n$). Thus, with probability at least $1-k^{-3.5}$, there are at most $9 \log k$ such vertices (rounding up conservatively), and thus $S \le 9 \log k$.
 
 In addition, we established that each vertex has at most $4$ neighbors in $U_a$. Thus, each vertex in $U_b$ adjacent to $U_a$ contributes at most $4$ edges to $E_{ab}$. Consequently, $|E_{ab}| \leq 4 \times 9 \log k = 36\log k$.
 
 A union bound over all $\binom{k}{2} < k^2$ pairs $a,b$ ensures this holds for all pairs simultaneously with probability at least $1 - k^2(k^{-3.5}) = 1 - k^{-1.5} \ge 0.99$.

 Finally, the sum $\sum_{a \neq b} \binom{|E_{ab} \cap C|}{2} \le \sum \frac{|E_{ab} \cap C|^2}{2}$ is maximized when $\frac{|C|}{36\log k}$ different sets $E_{ab}$ have the maximal size $36 \log k$. In that case:
 \[
  \text{Number of pairs} \le \frac{|C|}{36\log k} \cdot \frac{(36 \log k)^2}{2} =   18 |C| \log k.
 \]
Dividing by the total number of pairs $\frac{|C|(|C|-1)}{2}$, the fraction is at most $\frac{36 \log k}{|C|-1}$.
\end{proof}

\subsubsection{Construction of the point set}\label{sec:pointSet}

Let $E'$ be the set of all edges which are not bad.
We construct the point-set $P\subseteq \{0,1\}^{n+k}$, where for every $e\in E'$, we have a unique point $p_e\in P$. We identify the first $n$ coordinates with the set $V$ ($= [n]$) and the last $k$ coordinates with the set $[k]$. Let $e\in E'$ be incident on the vertices $i$ and $j$ in $V$.
Then we can define the point $p_e$ as follows:
$$
p_e:= \e_{i} + \e_{j}+\e_{\rho(i)+n}+\e_{\rho(j)+n},
$$
for any $j\in [n+k]$, $\e_j$ denotes the standard basis vector which is 1 on the $j^{\text{th}}$ coordinate and zero everywhere else.

Fix $e\in E'$. From the construction, we have:
\begin{align}
\forall e\in E',\ \|p_e\|_2^2=4.
\label{eq:norm}
\end{align}

We now have an upper bound on the diameter of the point-set we constructed. For any pair of points, we have:
\begin{align}
\forall{e,e'\in E'},\ \|p_e-p_{e'}\|_2^2\le \|p_e\|_2^2+\|p_{e'}\|_2^2\le 8
\label{eq:max}
\end{align}

For a pair of edges that intersect, we have that their corresponding pair of points have:
\begin{align}
2\le \|p_e-p_{e'}\|_2^2\le 4
\label{eq:inter}
\end{align}
Those at distance $2$ must have two colors and one vertex in common: from \cref{proposition:badVertex}, the number of such pairs is at most $100\frac{nd^2}{k}$, which is a tiny fraction of the total number of pairs.
For a pair of edges that don't intersect, their corresponding pair of points are at distance:
\begin{align}
4\le \|p_e-p_{e'}\|_2^2\le 8
\label{eq:nointer}
\end{align}

In the following sections, we abuse notation sometimes, and use $E$ instead of $E'$, whenever it is clear. In addition, since $|E'|=(1-o(1))\cdot |E|$ (as $k=\omega(1)$), we use $|E|$ in measuring clustering cost instead of $|E'|$.

From Propositions~\ref{prop:balanced}, \ref{proposition:badEdges}, \ref{proposition:badpairs}, and \ref{proposition:badVertex} and Lemma~\ref{lem:kappa_i}, with probability at least $0.95$, the fraction of bad edges is at most $100/k$, the fraction of non-representative pairs is  at most $600/k$, the number of edges that have one vertex and two colors in common with another edge is at most $100 \frac{nd^2}{k}$, and it holds that for every set $C$ of edges, the fraction of pairs of edges in $C$ having two colors in common is at most $\frac{36 \log k}{|C|}$. Under the event that the above happens, in Section~\ref{sec:completeness}, we show that  suppose there exists $n/2$ vertices in $G$ that covers $(1-\zeta)\cdot |E|$ edges the optimal $k$-means solution of $P$ has cost at most $3|E| - (1 - 7\zeta) kd$, in Section~\ref{sec:soundness}, we show that suppose every $n/2$ vertices in $G$ covers at most $(1-\delta)\cdot |E|$ edges then optimal $k$-means solution of $P$ has cost at least $3|E| - (1 - \delta^5) kd$.

\section{Completeness Analysis} \label{sec:completeness}

Suppose there exists  $n/2$ vertices in $G$  that covers the subset of edges $\tilde E$, where $|\tilde E|=(1-\zeta)\cdot |E|$. We show that, in that case, the cost of the optimal $k$-means solution in the instance constructed in \cref{sec:construction_k} is cheap, namely:
\[k\text{-means}(C_1,\ldots ,C_k)\leq 3|E| - (1 - 7\zeta) kd.\]

\subsection{Construction of the Clustering}
 Let $S = \{v_1, \ldots, v_{n/2}\}$ be the subset of vertices comprising the given cover.
This implies there exists a partition of $\tilde E$ into $n/2$ parts such that for all $i\in[n/2]$, all the edges in the $i^{\text{th}}$ part can be covered by $v_i$. Let us capture this specific partitioning more formally. There exists a function $\pi:\tilde E\to [n/2]$ such that for all $e\in \tilde E$ we have $v_{\pi(e)}$ covers $e$.
Let $\tilde P$ be the points (from $P$ as constructed in \cref{sec:pointSet}) corresponding to the edges $\tilde E$.
We define a partition of $\tilde P$ into $k$ clusters $\tilde C_1,\ldots ,\tilde C_k$ based on the colors of the cover vertices: the point $p_e$ belongs to cluster $\tilde C_i$ if and only if $\rho(v_{\pi(e)})=i$.
We extend this clustering to points in $E \setminus \tilde E$: for each edge $e \in E \setminus \tilde E$, arbitrarily designate one of its endpoints as $u_e$, and assign $p_e$   to the cluster $\overline C_{\rho(u_e)}$.
We finally define cluster $C_i := \tilde C_i \cup \overline C_i$.

We start by showing few properties about this partitioning. For a vertex $v_j$ from the vertex cover, define $d_j$ such that $d-d_j$ is the number of edges with $\pi(e) = j$. We have the following:

\begin{fact}\label{fact:dj}
 With the above notation, $\sum_{j=1}^{n/2} d_j = \zeta |E|$.
\end{fact}
\begin{proof}
 By definition of $d_j$, $v_j$ covers $d-d_j$ edges, therefore:
 \begin{align*}
  (1-\zeta)|E| = \sum_{j=1}^{n/2} (d-d_j) = \frac{nd}{2} - \sum_{j=1}^{n/2} d_j = |E| - \sum_{j=1}^{n/2} d_j.
 \end{align*}
 Thus, $\sum d_j = \zeta |E|$.
\end{proof}

We can also show that the $\tilde C_i$ are roughly balanced:
\begin{fact}\label{fact:sizeTCi}
 With probability $0.99$, each $\tilde C_i$ has size at least $\frac{nd}{6k}$.
\end{fact}
\begin{proof}
 Let $G = \{v_j \in S \mid d_j \le d/2\}$ be the set of ``good'' cover vertices.
 Since each vertex has degree $d$, $v_j$ covers at most $d$ edges, meaning $d_j \ge 0$.
 Applying Markov's inequality to the non-negative values $d_j$, and by Fact~\ref{fact:dj}, the number of vertices with $d_j > d/2$ is strictly less than
  $\frac{\zeta n d/2}{d/2} = \zeta n$.
 Consequently, the number of good vertices is bounded below by $|G| \ge n/2 - \zeta n = n(1/2 - \zeta)$.

 Since $\rho$ assigns each vertex in $V$ to a color in $[k]$ independently and uniformly at random, for a fixed color $i \in [k]$, the random variable $S_i := |G \cap U_i|$ follows a Binomial distribution $\text{Bin}(|G|, 1/k)$.
The expected value is $\mu := \E[S_i] = \frac{|G|}{k} \ge \frac{n(1/2 - \zeta)}{k}$. 
    
    From the preconditions of Theorem~\ref{thm:lb_k_n}, we know $\zeta \le \delta^5/10 \le 10^{-15}/10 < 1/100$, which implies $\mu \ge \frac{49n}{100k}$. We wish to bound the probability that $S_i < \frac{n}{3k}$. We apply the Chernoff bound for the lower tail (Theorem~\ref{thm:chernoff}), $\Pr(S_i \le (1-\lambda)\mu) \le \exp\left(-\frac{\lambda^2 \mu}{2}\right)$, with the deviation parameter $\lambda = 1/4$:
 \[
  (1-\lambda)\mu = \frac{3}{4} \mu \ge \frac{3}{4} \cdot \frac{49n}{100k} = \frac{147n}{400k}.
 \]
 Because $\frac{147n}{400k} = 0.3675 \frac{n}{k} > 0.3333 \frac{n}{k} \approx \frac{n}{3k}$, any event where $S_i < \frac{n}{3k}$   implies that $S_i \le \left(1-\frac{1}{4}\right)\mu$. Thus, we can safely upper-bound the probability:
 \[
  \Pr\left(S_i < \frac{n}{3k}\right) \le \Pr\left(S_i \le \left(1-\frac{1}{4}\right)\mu\right) \le \exp\left(-\frac{(1/4)^2 \mu}{2}\right) = \exp\left(-\frac{\mu}{32}\right) \le \exp\left(-\frac{49n}{3200k}\right).
 \]
 Since $k = o(|V|/\polylog(|V|))$, the ratio $n/k$ grows strictly faster than any polylogarithmic function, meaning $\exp(-49n/3200k)$ vanishes super-polynomially. Thus, for sufficiently large $n$, we have $\exp\left(-\frac{49n}{3200k}\right) \le \frac{0.01}{k}$.
 
 Applying a union bound over all $k$ clusters, we guarantee that $S_i \ge \frac{n}{3k}$ holds for all $i \in [k]$ simultaneously with probability at least $1 - k\left(\frac{0.01}{k}\right) = 0.99$.

 Conditioned on this event occurring, we can lower bound the size of each $\tilde C_i$. Recalling that $|\tilde C_i| = \sum_{v_j \in U_i \cap S} (d-d_j)$, we drop the non-negative contributions of vertices outside $G$ to obtain:
 \[
  |\tilde C_i| \ge \sum_{v_j \in G \cap U_i} (d-d_j) \ge \sum_{v_j \in G \cap U_i} \frac{d}{2} = \frac{d}{2} S_i \ge \frac{d}{2} \left(\frac{n}{3k}\right) = \frac{nd}{6k}.\qedhere
 \]
\end{proof}

\subsection{Cost of the Clustering}
The cost of cluster $C_i$ is $\frac{1}{2\left|C_i\right|}\cdot \left(\sum_{p_e,p_{e'}\in C_i}\|p_e-p_{e'}\|_2^2\right)$. Hence, we bound the distance between pairs of points in $C_i$. First, as all edges have at least one color in common (color $i$), they are at squared distance at most $6$ in the embedding.
Some are at squared distance $4$, when they share a vertex: we count more precisely their number, in terms of $d_j$:
\begin{fact}\label{fact:dist4pairs}
 In cluster $C_i$, the number of ordered pairs of edges at squared distance at most $4$ is at least
 \[ |\tilde C_i| \left( d - 1-\frac{6k}{n} \cdot \sum_{v_j \in S : \rho(v_j) = i} d_j \right). \]
\end{fact}
\begin{proof}
 First, the number of ordered pairs of edges in $C_i$ that have $v_j$ in common is $(d-d_j) (d-d_j - 1)$.
 Hence, the total number of ordered pairs of edges in $\tilde C_i$ with one vertex in common is :
 \begin{align*}
  \sum_{v_j \in S : \rho(v_j) = i} (d-d_j) (d-d_j - 1) &= \sum_{v_j \in S : \rho(v_j) = i} (d-d_j) (d-1) - \sum_{v_j \in S : \rho(v_j) = i} (d-d_j) d_j \\
  &= |\tilde C_i| (d-1) - \sum_{v_j \in S : \rho(v_j) = i} (d-d_j) d_j\\
  &\geq |\tilde C_i| (d-1) - d \sum_{v_j \in S : \rho(v_j) = i} d_j\\
  &\geq |\tilde C_i| \left( d -1 - \frac{6k}{n} \cdot \sum_{v_j \in S : \rho(v_j) = i} d_j \right),
 \end{align*}
 where the last inequality holds because $|\tilde C_i| \ge \frac{nd}{6k}$ (Fact~\ref{fact:sizeTCi}), implying $\frac{d}{|\tilde C_i|} \le \frac{6k}{n}$. As all those edges are at squared distance at most $4$, this concludes the claim.
\end{proof}

We can now compute the $k$-means cost.
\begin{proposition}\label{prop:kmeans_cost}
Suppose $S = \{v_1, \ldots, v_{n/2}\}$  covers the subset of edges $\tilde E$, where $|\tilde E|=(1-\zeta)\cdot |E|$. Then,
$$k\text{-means}(C_1,\ldots ,C_k)\leq 3|E| - (1 - 7\zeta) kd.$$
\end{proposition}

\begin{proof}
 Recall that the $k$-means cost is $\sum_{i=1}^k \frac{1}{2|C_i|} \sum_{e,e' \in C_i} \|p_e - p_{e'}\|^2$. The sum runs over all $|C_i|^2$ ordered pairs. The trivial distance bound is $\|p_e - p_{e'}\|^2 \le 6$ for all $e \neq e'$.
 Specifically, the cost contribution of cluster $C_i$ is at most:
 $$\frac{1}{2|C_i|} \sum_{e \neq e'} 6 = \frac{1}{2|C_i|} 6 |C_i|(|C_i|-1) = 3(|C_i|-1) = 3|C_i| - 3.$$
 Summing this base cost over all $k$ clusters gives $\sum_{i=1}^k (3|C_i| - 3) = 3|E| - 3k \le 3|E|$.
 
 However, pairs sharing a vertex have squared distance at most $4$. This is a reduction of at least $6-4=2$ per ordered pair compared to the base bound of $6$.
 Let $M_i$ be the number of unordered pairs in $C_i$ at distance 4. This corresponds to $2M_i$ ordered pairs. The cost reduction is   at least $\sum_{i=1}^k \frac{1}{2|C_i|} (2 \cdot 2M_i) = \sum_{i=1}^k \frac{2M_i}{|C_i|}$.
 
 Using Fact~\ref{fact:dist4pairs}, let $\Delta_i = \sum_{v_j \in S : \rho(v_j) = i} d_j$. We have $M_i \ge \frac{1}{2} |\tilde C_i| \left( d - 1 - \frac{6k}{n} \Delta_i \right)$.
 Since $|C_i| = |\tilde C_i| + |\overline C_i|$, we can write the fraction $\frac{|\tilde C_i|}{|C_i|}$ as $\frac{1}{1 + |\overline C_i|/|\tilde C_i|}$.
 
 The total cost reduction is bounded below by:
 $$\sum_{i=1}^k \frac{|\tilde C_i|\left(d - 1 - \frac{6k}{n} \Delta_i\right)}{|C_i|} = \sum_{i=1}^k \frac{d - 1 - \frac{6k}{n} \Delta_i}{1 + \frac{|\overline C_i|}{|\tilde C_i|}}.$$
 
 Using the algebraic identity $\frac{B-C}{1+A} \ge B - C - AB$ (which is valid for all $A, B, C \ge 0$ since $\frac{B-C}{1+A} - (B - C - AB) = \frac{AC + A^2 B}{1+A} \ge 0$), we set $A = \frac{|\overline C_i|}{|\tilde C_i|}$, $B = d-1$, and $C = \frac{6k}{n} \Delta_i$. We obtain:
 $$\sum_{i=1}^k \frac{d - 1 - \frac{6k}{n} \Delta_i}{1 + \frac{|\overline C_i|}{|\tilde C_i|}} \ge \sum_{i=1}^k \left[ (d-1) - \frac{6k}{n}\Delta_i - (d-1) \frac{|\overline C_i|}{|\tilde C_i|} \right]\ge \sum_{i=1}^k \left[ (d-1) - \frac{6k}{n}\Delta_i - d \frac{|\overline C_i|}{|\tilde C_i|} \right].$$
 
 We bound each term in the summation. First we have  $\sum_{i=1}^k (d-1) = k(d-1) = kd - k$. Second, we have  $\sum_{i=1}^k \frac{6k}{n} \Delta_i = \frac{6k}{n} \sum_i \Delta_i = \frac{6k}{n} (\zeta |E|) = \frac{6k}{n} \frac{\zeta n d}{2} = 3\zeta kd$. (Using Fact~\ref{fact:dj}). Finally, we have from Fact~\ref{fact:sizeTCi}, $|\tilde C_i| \ge \frac{nd}{6k}$, so $\frac{1}{|\tilde C_i|} \le \frac{6k}{nd}$. Thus:
  $$\sum_{i=1}^k d \frac{|\overline C_i|}{|\tilde C_i|} \le d \frac{6k}{nd} \sum_{i=1}^k |\overline C_i| = \frac{6k}{n} (\zeta |E|) = 3\zeta kd.$$
 
 Combining these, the total cost reduction is at least:
 $$(kd - k) - 3\zeta kd - 3\zeta kd = kd(1 - 6\zeta) - k.$$
 
 The total cost is therefore at most:
 \[(3|E| - 3k) - (kd(1 - 6\zeta) - k) = 3|E| - kd(1 - 6\zeta) - 2k \le 3|E| - kd(1 - 7\zeta).\qedhere\]
\end{proof}

\section{Soundness Analysis} \label{sec:soundness}
Let the soundness assumption be that for every $S\subseteq V$ of size $n/2$, we have that there exists $F\subseteq E$ such that $|F|\ge \delta\cdot |E|$ and for all $e\in F$ we have $e\cap S=\emptyset$.
Consider an arbitrary partition of $P$ into $k$ clusters $C_1,\ldots ,C_k$. If this clustering has cost at most  $3|E| - (1-\delta^5)dk$ then we will show that the soundness assumption is contradicted\footnote{Sometimes in the analysis, for the sake of ease of presentation, we bring back the bad edges removed from the construction of the point-set. This helps with using the regularlity of the graph. However, this doesn't affect the conclusion, as $k > \sqrt{n}d$ implies $|E_{\text{bad}}| \le \frac{50nd}{k} < \frac{50k}{d} = o(k) \ll \tau kd$. Hence their impact on the clustering cost is negligble.}.

For every $i\in[k]$, $G_i(V_i,E_i\subseteq E)$ is the subgraph defined as follows: $e\in E_i$ if and only if $p_e\in C_i$ (and $G_i$ has no isolated vertices).
In words, $E_i$ is the set of edges whose embedded point is in cluster $C_i$ and thus $|E_i| = |C_i|$.

We start by removing all pairs of edges with one vertex and two colors in common. From \Cref{proposition:badVertex}, this is at most $\frac{100nd^2}{k} = o(kd)$ many edges (as $k>\sqrt{n}$ and $d=(\log n)^{L}$).
As this is a tiny fraction of the total number of edges, not covering them does not hurt the soundness analysis, and it only decreases the $k$-means cost.

In addition, we assume that every cluster is of size at least $\xi d$ for some $\xi>0$. This is again because the sum of points in all clusters of size less than $\xi d$ is smaller than $\xi dk$ overall, and thus even removing them does not affect our analysis (assuming $\xi$ is sufficiently small, for example, $\xi=2^{-1/\delta}$).

Our proof strategy can now be broken down to the following four steps.

First, we connect the cost of cluster $C_i$ (for $i\in [k]$) with certain properties of the graph $G_i$. Formally, we show in \Cref{sec:costCluster} the following:
\begin{restatable}{lemma}{costCluster}\label{lem:costCluster}
 Fix a cluster $C_i$ (for $i\in[k]$). 
 Let $\gamma_i$ be the fraction of pairs of edges $e, e' \in E_i$ that have no color in common, and $\kappa_i$ be the fraction of pairs of edges that have two colors in common. Then, the cost of the cluster is equal to
 \begin{align*}
 3|C_i|-1+ (\gamma_i - \kappa_i) (|C_i|-1)-\frac{1}{|C_i|}\cdot \sum_{v \in V_i} d_{i,v}^2,
\end{align*}
where $d_{i,v}$ is the degree of $v$ in the graph $G_i$.
\end{restatable}

Next, we show that if $\gamma_i$ (fraction of pairs of edges with no color in common) and $\kappa_i$ (fraction of pairs of edges with two colors in common) are not too large, then there is a dominant single color in the cluster $C_i$. Formally, we show in \Cref{sec:monochrome} the following:
\begin{restatable}{lemma}{monochrome}\label{lem:monochrome}
Let $1/1000 > \eta > 0$ be some constant to be specified later.
Fix a cluster $C_i$ (for some $i\in[k]$).
If $\kappa_i \leq \eta^3$ and $\gamma_i \leq \eta/3$, then there is a set of $(1-\eta)$ fraction of edges in $E_i$ that have the same color $c_i$.
\end{restatable}

Then, we show that we can identify a large subcollection of clusters for which both $\kappa_i$ and $\gamma_i$ are small. Moreover, for each cluster in this subcollection, we obtain a meaningful bound on the sum of the squared degrees of the vertices in the subcollection, and this bound is tied to the clustering cost appearing in Lemma~\ref{lem:costCluster}. Formally, we show in \Cref{sec:avgStruct} the following:

\begin{restatable}{lemma}{avgStruct}\label{lem:avgStruct}     
    There is some $I\subseteq [k]$ such that the following holds: 
\begin{enumerate}
    \item \label{item:gamma} for all $i\in I$, $\kappa_i \leq \eta^3$ and $\gamma_i \leq \eta/3$.
    \item \label{item:lotEdges}$\sum_{i\in I}|E_i|\ge (1-38\delta^2)\cdot |E|$.
\end{enumerate}
In addition, we also have for all $i\in I$:
    \begin{align}
 \frac{1}{|E_i|}\cdot \sum_{v\in G_i} d_{i,v}^2 \geq (1-\delta^3) d,
\label{eqgood}
\end{align}
\end{restatable}

Finally, we show that a cost of clustering $3|E|-(1-\delta^5)dk$ implies that we can construct a set of vertices $S\subseteq [n]$ of size  $n/2$ such that we cover most of the edges. 

\begin{restatable}{lemma}{smallVC}\label{lem:smallVC}
For sufficiently large $n$, if the clustering cost is at most $3|E| - (1-\delta^5)kd$, 
then by setting $\eta=\delta^5$, there exists a subset of exactly $n/2$ vertices that covers at least $1-10\delta^{1.5}$ fraction of $E$. 
\end{restatable}

 For small enough $\delta$, this implies our constructed $n/2$-sized subset misses strictly fewer than $\delta$ fraction of the edges contradicting the soundness assumption.

\subsection{Proof of Lemma~\ref{lem:costCluster}: Connecting Cluster Cost to Properties of Graph}\label{sec:costCluster}
\costCluster*
\begin{proof}
 First note that from Lemma~\ref{lem:folklore}:
 \begin{align}
  \text{Cost}(C_i) = \frac{1}{2|C_i|} \sum_{e \in E_i} \sum_{e' \in E_i} \|p_e - p_{e'}\|^2 = \sum_{e \in E_i} \|p_e\|^2 - \frac{1}{|C_i|} \left\| \sum_{e \in E_i} p_e \right\|^2.\label{eqcostmain}   
 \end{align}
  
 Since $\|p_e\|^2 = 4$ for all $e$, the first term is $4|C_i|$.
 
 Now consider the vector sum $S = \sum_{e \in E_i} p_e$. Recall $p_e = \e_u + \e_v + \e_{\rho(u)+n} + \e_{\rho(v)+n}$.
 We can write $S$ in terms of the basis vectors. The coordinate $\e_v$ appears $d_{i,v}$ times (once for each incident edge in $E_i$). The coordinate $\e_{c+n}$ appears $\kappa_{i,c}$ times, where $\kappa_{i,c}$ is the number of edges in $E_i$ that have an endpoint of color $c$.
 Thus:
 \[
  \|S\|^2 = \sum_{v \in V_i} d_{i,v}^2 + \sum_{c \in [k]} \kappa_{i,c}^2.
 \]
 The term $\sum_{c} \kappa_{i,c}^2$ counts the number of pairs $(e,e')$ that share a specific color $c$, summed over all colors.
 \[
  \sum_{c} \kappa_{i,c}^2 = \sum_{c} \left( \sum_{e \in E_i} \mathbbm{1}_{c \in \rho(e)} \right)^2 = \sum_{e,e' \in E_i} \sum_{c} \mathbbm{1}_{c \in \rho(e) \cap \rho(e')} = \sum_{e,e'} |\rho(e) \cap \rho(e')|.
 \]
 We split the sum into $e=e'$ and $e \neq e'$.
 For $e=e'$, we have $|\rho(e) \cap \rho(e')| = 2$. There are $|C_i|$ such terms and thus the total sum contribution is $2|C_i|$.
 For $e \neq e'$, let $\gamma_i$ be the fraction with 0 colors common, and $\kappa_i$ be the fraction with 2 colors common. The remaining fraction $(1-\gamma_i-\kappa_i)$ have exactly 1 color common.
 The sum is:
 \[
  \sum_{e \neq e'} |\rho(e) \cap \rho(e')| = |C_i|(|C_i|-1) \left[ 0 \cdot \gamma_i + 1 \cdot (1-\gamma_i-\kappa_i) + 2 \cdot \kappa_i \right] = |C_i|(|C_i|-1) (1-\gamma_i+\kappa_i).
 \]
 Thus, $\|S\|^2 = \sum d_{i,v}^2 + 2|C_i| + |C_i|(|C_i|-1)(1-\gamma_i+\kappa_i)$.
 Substituting back into \eqref{eqcostmain}, we have:\allowdisplaybreaks
 \begin{align*}
  \text{Cost}(C_i) &= 4|C_i| - \frac{1}{|C_i|} \left( \sum_{v \in V_i} d_{i,v}^2 + 2|C_i| + |C_i|(|C_i|-1)(1-\gamma_i+\kappa_i) \right) \\
  &= 4|C_i| - \frac{1}{|C_i|}\sum d_{i,v}^2 - 2 - (|C_i|-1)(1-\gamma_i+\kappa_i) \\
  &= 4|C_i| - 2 - \left[ (|C_i|-1) - (|C_i|-1)\gamma_i + (|C_i|-1)\kappa_i \right] - \frac{1}{|C_i|}\sum d_{i,v}^2 \\
  &= 4|C_i| - 2 - |C_i| + 1 + (|C_i|-1)(\gamma_i-\kappa_i) - \frac{1}{|C_i|}\sum d_{i,v}^2 \\
  &= 3|C_i| - 1 + (|C_i|-1)(\gamma_i-\kappa_i) - \frac{1}{|C_i|}\sum d_{i,v}^2. \qedhere
 \end{align*}
\end{proof}

\subsection{Proof of Lemma~\ref{lem:monochrome}: Typical Clusters are Monochromatic}\label{sec:monochrome}

\monochrome*

\begin{proof}
Assume towards contradiction that there is no color intersecting more than a $(1-\eta)$ fraction of the edges.

Let $f_c$ be the fraction of edges intersecting color $c$. We first provide an upper-bound on the probability that two distinct edges $e,e'$ drawn uniformly at random overlap in at least one color.
Let $\tilde f_c = \frac{f_c |C_i| - 1}{|C_i|-1} \leq f_c$: for a pair of distinct edges $e, e'$, $\Pr[c \in \rho(e') \mid c \in \rho(e)] = \tilde f_c \leq f_c$.
Therefore, the probability of overlap is at most:
\begin{align*}
 \Pr_{e \neq e'}[\exists c \in \rho(e) \cap \rho(e')] &\leq \sum_{c} \Pr_{e \neq e'}[c \in \rho(e) \text{ and } c \in \rho(e')] \leq \sum_{c} f_c^2.
\end{align*}
Since the probability of overlap is $1-\gamma_i$, we get
\begin{equation}
\label{eq:fraction}
 1-\eta/3 \leq 1-\gamma_i \leq \sum_{c} f_c^2.
\end{equation}

Since each edge gets exactly $2$ colors, $\sum f_c = 2$. Order the colors by decreasing frequency, $f_1 \geq f_2 \geq \dots$. Let $f_{j, \ell}$ be the fraction of edges with colors $j$ and $\ell$.

Note that $\kappa_i$ is exactly the fraction of distinct pairs of edges sharing two colors. Let $X_{j,\ell} = f_{j,\ell}|C_i|$ be the number of edges with colors $j$ and $\ell$. We have $\sum_{j<\ell} \frac{X_{j,\ell}(X_{j,\ell}-1)}{|C_i|(|C_i|-1)} = \kappa_i$.
Since $\frac{X(X-1)}{N(N-1)} \ge \frac{X^2 - X}{N^2} = f^2 - \frac{f}{N}$, we obtain $f_{j,\ell}^2 \le \kappa_i + \frac{f_{j,\ell}}{|C_i|} \le \kappa_i + \frac{1}{|C_i|}$.
Since we pruned clusters smaller than $\xi d$, and $d = \omega(1)$, the term $1/|C_i|$ vanishes for large $n$. Thus, we can define $\tilde{\kappa}_i := \kappa_i + 1/|C_i| \le 2\eta^3$, yielding $f_{j,\ell} \le \sqrt{\tilde{\kappa}_i}$.

By the inclusion-exclusion principle, for any $t$ colors, the indicator variables satisfy $\sum_{j \in [t]} f_j - \sum_{j < \ell \in [t]} f_{j, \ell} \leq 1$, which implies:
\begin{equation} \label{eq:incl-excl}
 \sum_{j=1}^t f_j \leq 1 + \binom{t}{2} \sqrt{\tilde{\kappa}_i}.
\end{equation}

We seek to maximize $\sum f_c^2$ under these constraints. The sum is maximized when the mass is concentrated on as few colors as possible.
Define $L = \{c \ge 2 : f_c = f_2\}$. As long as $f_1 \neq 1-\eta$ and $|L| \le 3$, continuously transfer an arbitrarily small mass $\varepsilon>0$ from each color in $L$ to $f_1$. Update $L$ by adding colors whose frequency drops to equal the next largest frequency.
This strictly increases $\sum f_c^2$ (by convexity), preserves the non-increasing order, and maintains the partial sum $\sum_{j=1}^{1+|L|} f_j$,   satisfying \cref{eq:incl-excl} for $t = 1+|L|$.

\paragraph{Case 1: $f_1 = 1-\eta$ and $|L| \le 3$.}
Applying \cref{eq:incl-excl} for $t=1+|L|$ ensures $f_1 + |L|f_2 \le 1 + \binom{|L|+1}{2}\sqrt{\tilde{\kappa}_i}$. Substituting $f_1 = 1-\eta$:
\[ f_2 \le \frac{\eta}{|L|} + \frac{|L|+1}{2}\sqrt{\tilde{\kappa}_i}. \]
Since $1 \le |L| \le 3$, the maximum occurs at $|L|=1$, giving $f_2 \le \eta + 2\sqrt{\tilde{\kappa}_i}$.
The sum of squares is  upper bounded by concentrating the remaining mass on $f_2$:
\begin{align*}
 \sum f_i^2 &\le f_1^2 + f_2(2-f_1) \le (1-\eta)^2 + (\eta + 2\sqrt{\tilde{\kappa}_i})(1+\eta).
\end{align*}
Since $\tilde{\kappa}_i \le 2\eta^3$, we have $\sqrt{\tilde{\kappa}_i} \le 1.5\eta^{1.5}$. The sum expands to:
\[
1 - 2\eta + \eta^2 + \eta + \eta^2 + 3\eta^{1.5}(1+\eta) \le 1 - \eta + 2\eta^2 + 6\eta^{1.5}.
\]
For $\eta \le 1/1000$, the quantity $2\eta^2 + 6\eta^{1.5}$ is  less than  $\frac{2}{3}\eta$. Thus, the sum is strictly less than $1 - \eta/3$, contradicting \cref{eq:fraction}.

\paragraph{Case 2: $|L| = 4$.}
If the procedure stops because $|L|$ reaches 4, then $f_2 = f_3 = f_4 = f_5$. Evaluating \cref{eq:incl-excl} for $t=5$ implies $f_1 + 4f_2 \leq 1+ \binom{5}{2}\sqrt{\tilde{\kappa}_i} = 1 + 10\sqrt{\tilde{\kappa}_i}$.
Hence, $f_2 \leq \frac{1+10\sqrt{\tilde{\kappa}_i} - f_1}{4}$.
The sum of squares is bounded by:
\begin{align*}
 \sum f_i^2 &\leq f_1^2 + 2 f_2 \leq f_1^2 + \frac{1-f_1}{2} + 5\sqrt{\tilde{\kappa}_i}.
\end{align*}
The quadratic function which maps $x$ to $x^2 - x/2$ is maximized on $[1/4, 1-\eta]$ at the boundary $x=1-\eta$. The value is bounded by:
\[ (1-\eta)^2 + \frac{\eta}{2} + 5(1.5\eta^{1.5}) = 1 - 1.5\eta + \eta^2 + 7.5\eta^{1.5}. \]
We require this to be less than $1 - \eta/3$, which implies $1.16\eta > \eta^2 + 7.5\eta^{1.5}$.
Dividing by $\eta$, we need $1.16 > \eta + 7.5\sqrt{\eta}$. For $\eta \le 1/1000$, $\sqrt{\eta} \approx 0.0316$, making the right side roughly $0.238 < 1.16$.
This provides the final contradiction.
\end{proof}

\subsection{Proof of Lemma~\ref{lem:avgStruct}: Identifying large monochromatic clusters}\label{sec:avgStruct} 
\avgStruct*

For the sake of presentation, we use  $\tau := \delta^5$ and $\omega:=\delta^3$.  Moreover, recall that $\alpha = \delta^{10}$, and we set $\eta = \delta^5$.

We first apply a pre-processing step to discard small clusters. We discard all clusters satisfying $|E_i| < \frac{10}{\tau^2} d$. The total number of edges discarded is   bounded by $k \cdot \frac{10}{\tau^2} d = \frac{10}{\tau^2} kd$. Because $k = o(n)$ and $|E| = nd/2$, this quantity is   $o(|E|)$. Thus, removing them does not significantly affect the total edge mass. From now on, we assume $|E_i| \ge \frac{10}{\tau^2} d = 10\delta^{-10} d$ for all remaining clusters. 

Because $d = (\log n)^{L}$, for sufficiently large $n$, this size guarantees that $\frac{36 \log k}{|E_i|-1} \le \eta^3$. Thus, from Lemma~\ref{lem:kappa_i}, $\kappa_i \le \eta^3$ is  satisfied for all surviving clusters.

For readability, we break the proof  of Lemma~\ref{lem:avgStruct} in the following subsections.

\subsubsection{Large Clusters are Bad}
\begin{lemma}\label{lem:giant_cluster_strict}
    Let $\mathcal{C}=\{C_1, \dots, C_k\}$ be any clustering of the point set $P$. Suppose there exists a cluster $C \in \mathcal{C}$ such that $|C| \ge 16kd$. 
    Assuming $n$ is sufficiently large, the cost of the clustering is greater than $3|E| - (1 - \delta^5)kd$.
\end{lemma}

\begin{proof}
    For any cluster $A \in \mathcal{C}$, define its relative benefit as $\Delta(A) = 3|A| - \text{Cost}(A)$. Using Lemma~\ref{lem:costCluster}:
    \[
        \Delta(A) = 1 - (\gamma_A - \kappa_A)(|A|-1) + \frac{1}{|A|}\sum_{v \in V_A} d_{A,v}^2.
    \]
    Since the maximum degree in the induced subgraph is $d$, $\frac{1}{|A|}\sum d_{A,v}^2 \le d \cdot \frac{1}{|A|}\sum d_{A,v} = \frac{d(2|A|)}{|A|} = 2d$.
    From Lemma~\ref{lem:kappa_i}, $\kappa_A(|A|-1) \le 36 \log k$. Because $\gamma_A \ge 0$, we have $-(\gamma_A - \kappa_A)(|A|-1) \le 36 \log k$.
    Thus,   $\Delta(A) \le 1 + 36 \log k + 2d$.
    Summing over all $k$ clusters, the maximum possible benefit is $\mathcal{B}_{all} \le k(2d + 36 \log k + 1) \le 3kd$.

    Now, consider the giant cluster $C$ with $|C| \ge 16kd$. Since $k > \sqrt{n}$, we have $k^2 > n \implies kd > nd/k$. Thus, $|C| \ge 16\frac{nd}{k}$.
    Let $f_c$ be the fraction of edges in $C$ incident to color $c$. From Proposition~\ref{prop:balanced}, the absolute maximum number of edges incident to color $c$ is at most $d|U_c| \le \frac{1.1nd}{k}$.
    Thus, the maximum frequency is $\mu = \max_c f_c \le \frac{1.1nd/k}{|C|} \le \frac{1.1nd/k}{16nd/k} = \frac{1.1}{16} < 0.07$.
    The fraction of pairs sharing at least one color is bounded by $1-\gamma_C \le \sum f_c^2 \le \mu \sum f_c = 2\mu \le 0.14$.
    Thus, $\gamma_C \ge 0.86$. Similarly, $\kappa_C \le 1 - \gamma_C \le 0.14$.
    Therefore, the penalty coefficient is firmly bounded: $\gamma_C - \kappa_C \ge 0.86 - 0.14 = 0.72 \ge \frac{1}{2}$.
    
    The benefit of the giant cluster is   bounded by:
    \[
        \Delta(C) \le 1 - \frac{1}{2}(|C|-1) + 2d \le 2d + 1.5 - \frac{|C|}{2}.
    \]
    We refine our global bound on the total benefit $\mathcal{B}$ by using this penalty for $C$ alongside the   upper bound for the remaining $k-1$ clusters:
    \[
        \mathcal{B} \le \Delta(C) + \sum_{A \neq C} \Delta(A) \le 2d + 1.5 - \frac{|C|}{2} + 3kd.
    \]
    Since $|C| \ge 16kd$, we have $-\frac{|C|}{2} \le -8kd$.
    Thus, $\mathcal{B} \le 3kd + 2d + 1.5 - 8kd < -4kd < 0$.
    Consequently, the total cost is at least $ 3|E| - \mathcal{B} > 3|E| > 3|E| - (1-\delta^5)kd$.
\end{proof}

\subsubsection{Universal Expansion Bound on Squared Degrees}
Before evaluating the global cost, we establish a   upper bound on $S_i := \frac{1}{|E_i|}\sum_{v \in V_i} d_{i,v}^2$ for {every} surviving cluster $i$. 
We use the approach from \cite{Cohen-AddadS19}. Consider the set of low-degree vertices $W_i = \{v \in V_i \mid d_{i,v} \leq \sqrt{\alpha} d\}$. The contribution of these vertices to the sum of squares is bounded by their maximum degree: $\sum_{v \in W_i} d_{i,v}^2 \leq \sqrt{\alpha} d \sum_{v \in W_i} d_{i,v} \le 2 \sqrt{\alpha} d |E_i|$.

Now, consider the set $W'_i = V_i \setminus W_i$ of vertices with $d_{i,v} > \sqrt{\alpha} d$. By the contrapositive of Lemma~\ref{lem:giant_cluster_strict}, our global cost assumption $\Phi(\mathcal{C}) \le 3|E| - (1-\tau)kd$ guarantees that no cluster exceeds size $16kd$. 
Consequently, the number of high-degree vertices in the induced subgraph is   bounded by $|W'_i| \le \frac{2|E_i|}{\sqrt{\alpha}d} \le \frac{32k}{\sqrt{\alpha}}$. 
Because $k = o(|V|/d^{\omega(1)})$, and $d^{\omega(1)}$ grows strictly faster than any polylogarithmic function, this guarantees that $|W'_i| \ll n / \polylog n$, which is  small enough to   apply the $(\polylog n, \alpha)$-small set vertex expansion property of Hypothesis~\ref{hypo:xxh}.

By expansion, the number of internal edges evaluated in the base graph $G$ is bounded by $2e_G(W'_i) \le \alpha d |W'_i| + |W'_i|$. 
The sum of degrees inside $G_i$ for these vertices is exactly twice the internal edges plus the edges crossing to $W_i$ (which is  bounded by $|E_i|$). Thus:
\[ \sum_{v \in W'_i} d_{i,v} \leq \alpha d |W'_i| + |W'_i| + |E_i| \leq \alpha d \left( \frac{2|E_i|}{\sqrt{\alpha}d} \right) + \frac{2|E_i|}{\sqrt{\alpha}d} + |E_i| \le (1+3\sqrt{\alpha})|E_i|, \]
where the final inequality holds for large $d$.
Their contribution to the squares is $\sum_{v \in W'_i} d_{i,v}^2 \le d \sum_{v \in W'_i} d_{i,v} \le d(1+3\sqrt{\alpha})|E_i|$.
Combining both sets, we conclude that for {every} cluster:
\begin{equation}
    S_i = \frac{1}{|E_i|}\sum_{v \in V_i} d_{i,v}^2 \leq 2\sqrt{\alpha}d + d(1 + 3\sqrt{\alpha}) \le d(1+5\sqrt{\alpha}).
    \label{eq:inI}
\end{equation}

\subsubsection{Restricting to clusters with small $\gamma_i$}
\begin{claim}\label{claim:drop_gamma}
    Let $B_\gamma = \{ i \in [k] \mid \gamma_i > \eta/3 \}$ be the set of clusters with high color disagreement. The total number of edges in $B_\gamma$ is  bounded by $\sum_{B_\gamma} |E_i| \le 22 kd = o(|E|)$.
\end{claim}
\begin{proof}
    Summing the cost formula from Lemma~\ref{lem:costCluster} over all clusters gives:
    \[
        3|E| - (1-\tau)kd \ge 3|E| - k + \sum_{i=1}^k \gamma_i(|E_i|-1) - \sum_{i=1}^k \kappa_i(|E_i|-1) - \sum_{i=1}^k S_i.
    \]
    Rearranging to isolate the $\gamma_i$ terms yields:
    \[
        \sum_{i=1}^k \gamma_i(|E_i|-1) \le \tau kd - kd + k + \sum_{i=1}^k S_i + \sum_{i=1}^k \kappa_i(|E_i|-1).
    \]
    Substituting our   bound $\sum S_i \le kd(1+5\sqrt{\alpha})$ and $\sum \kappa_i(|E_i|-1) \le 36k\log k$:
    \[
        \sum_{i=1}^k \gamma_i(|E_i|-1) \le \tau kd + 5\sqrt{\alpha}kd + k + 36k\log k.
    \]
   
   By our choice of $\tau = \delta^5$ aligning with $\sqrt{\alpha} = \delta^5$, the right-hand side is   bounded by $\tau kd + 5\tau kd + \tau kd = 7\tau kd$.
    Restricting the sum to $B_\gamma$ (where $\gamma_i > \eta/3$) yields $\frac{\eta}{3} \sum_{i \in B_\gamma} (|E_i|-1) \le 7\tau kd$. 
    Because we defined $\eta = \tau$, this simplifies to:
    \[ \sum_{i \in B_\gamma} |E_i| \le 21 kd + |B_\gamma| \le 22 kd = o(|E|). \qedhere \]
\end{proof}

\subsubsection{Identifying clusters with optimal degree sums}
We return to the clustering cost equation to bound the clusters with poor degree concentration. 
Let $I_{bad} = \{i \in [k] \mid S_i < d(1-\omega)\}$. Since Equation \eqref{eq:inI} establishes $S_i \le d(1+5\sqrt{\alpha})$   for all $k$ clusters, we cleanly partition the global sum into $I_{bad}$ and $[k] \setminus I_{bad}$:
\[ \sum_{i=1}^k S_i = \sum_{I_{bad}} S_i + \sum_{[k] \setminus I_{bad}} S_i \le |I_{bad}| d(1-\omega) + (k - |I_{bad}|) d(1+5\sqrt{\alpha}). \]
Isolating $\sum S_i$ from the initial cost equation gave $\sum_{i=1}^k S_i \ge kd(1-1.1\tau)$. 
Substituting the lower bound yields:
\[ kd(1-1.1\tau) \le |I_{bad}| d(1-\omega) + (k - |I_{bad}|) d(1+5\sqrt{\alpha}). \]
Dividing by $kd$ and letting $x = |I_{bad}|/k$ be the fraction of such clusters, we have:
\[ 1 - 1.1\tau \le x(1-\omega) + (1-x)(1+5\sqrt{\alpha}) = 1 + 5\sqrt{\alpha} - x(\omega + 5\sqrt{\alpha}). \]
Subtracting 1 and isolating the $x$ terms to the left:
\[ x(\omega + 5\sqrt{\alpha}) \le 1.1\tau + 5\sqrt{\alpha}. \]
Since $\omega = \delta^3$, $\tau = \delta^5$, and $\sqrt{\alpha} = \delta^5$, we have:
\[ x(\delta^3 + 5\delta^5) \le 1.1\delta^5 + 5\delta^5 = 6.1\delta^5. \]
Thus, $x \le \frac{6.1\delta^5}{\delta^3} = 6.1\delta^2$.
This firmly bounds the total number of sub-optimal clusters as $|I_{bad}| \le 6.1\delta^2 k$.

\subsubsection{Bounding the edges of the good clusters}
We define our final good set of clusters as $I = [k] \setminus (P_{small} \cup B_\gamma \cup I_{bad})$ (where $P_{small}$ is the set of pruned clusters). By definition, every $i \in I$ inherently satisfies $\gamma_i \le \eta/3$, $\kappa_i \le \eta^3$, and $S_i \ge d(1-\omega)$.

We show Item \ref{item:lotEdges}, namely that most edges are safely retained in $I$. 
The omitted edges belong   to $P_{small} \cup B_\gamma \cup I_{bad}$. 
We established that $\sum_{P_{small}} |E_i| = o(|E|)$, and $\sum_{B_\gamma} |E_i| \le 22kd = o(|E|)$.
For the clusters in $I_{bad} \setminus B_\gamma$, we leverage the fact that they   satisfy $\gamma_i \le \eta/3$. 

Fix any such cluster $i$. 
By Proposition~\ref{prop:balanced} with high probability, no color is assigned to more than $1.1 n/k$ vertices. The maximum number of edges incident to any single color $c$ is therefore at most $1.1 nd/k$.
Let $e \in E_i$ with colors $\rho(e) = \{c_1, c_2\}$. The number of edges in $E_i$ sharing at least one color with $e$ is at most $2.2nd/k$.
Thus, $e$ shares absolutely no colors with the remaining subset of edges, which has size at least $|E_i| - \frac{2.2nd}{k}$.
Summing this disjointness over all $e \in E_i$ evaluates each completely disjoint pair exactly twice. The total number of disjoint pairs is therefore at least $\frac{1}{2} |E_i| \left( |E_i| - \frac{2.2nd}{k} \right)$.
Dividing by the total number of valid pairs $\binom{|E_i|}{2} \le \frac{1}{2}|E_i|^2$, we get:
\[ \gamma_i \geq \frac{\frac{1}{2}|E_i|(|E_i| - 2.2nd/k)}{\frac{1}{2}|E_i|^2} = 1 - \frac{2.2nd}{k|E_i|}. \]
If $|E_i| > \frac{3nd}{k}$, then $\gamma_i \ge 1 - \frac{2.2}{3} \ge 0.26$. However, we know $\gamma_i \leq \eta/3 < 0.01$, leading to a contradiction. 

Thus, every cluster in $I_{bad} \setminus B_\gamma$ must deterministically satisfy $|E_i| \le \frac{3nd}{k}$. 
Since there are at most $|I_{bad}| \le 6.1\delta^2 k$ such clusters, their absolute maximum edge mass is  bounded by:
\[ 6.1\delta^2 k \cdot \frac{3nd}{k} = 18.3\delta^2 nd = 36.6\delta^2 |E|. \]
Adding the limits together, the total number of edges omitted from $I$ is  bounded by $36.6\delta^2|E| + o(|E|) \le 38\delta^2|E|$ for sufficiently large $n$. Thus, $\sum_{i \in I} |E_i| \ge (1-38\delta^2)|E|$.

\subsection{A cheap clustering admits a small vertex cover}\label{sec:smallVC}
\smallVC*
\begin{proof}
In any   cluster $C_i \in I$, Lemma~\ref{lem:monochrome} guarantees the existence of a dominant color $c_i$ shared by at least a $(1-\eta)$ fraction of the edges in $E_i$. Let $\tilde{E}_i \subseteq E_i$ be this subset of edges. 

Let $\tilde{V}_i \subseteq V_i$ be the set of endpoints of edges in $\tilde{E}_i$ that have color $c_i$. Since bad edges (monochromatic endpoints) were removed from $P$ in Section~\ref{sec:construction_k}, no edge connects two vertices of the same color. Thus, $\tilde{V}_i$ is   an independent set in $G$, and each edge in $\tilde{E}_i$ has exactly one endpoint in $\tilde{V}_i$. 

Before proceeding, we apply the  secondary pruning step: let $I' \subseteq I$ be the subset of good clusters satisfying $|E_i| \ge \tau \frac{nd}{k}$. Discarding the clusters with $|E_i| < \tau \frac{nd}{k}$ removes at most $k \cdot \tau \frac{nd}{k} = \tau n d = 2\tau|E|$ edges in total. Thus, the pruned set $I'$ safely retains at least $(1-38\delta^2 - 2\tau)|E|$ edges. Because $\tau = \delta^5 \ll \delta^2$, we cleanly bound this retained mass by $(1-39\delta^2)|E|$.

Fix $i\in I'$. Let $S_i$ be the neighborhood of $\tilde{V}_i$ in $G_i$. Let $\tilde S_i\subseteq S_i$ be the set of unique neighbors of $\tilde{V}_i$ in $G_i$, meaning each vertex in $\tilde S_i$ has exactly one neighbor in $\tilde V_i$ within $G_i$.

\begin{claim}\label{claim:S_i}
    $|\tilde{S}_i|\ge |\tilde E_i|-2\alpha d |\tilde{V}_i|$.
\end{claim}
\begin{proof}
Let $R\subseteq V$ be the set of vertices in the graph $G$ which have more than one neighbor in $\tilde V_i$. Let $m_R$ be the number of edges between $\tilde V_i$ and $R$ in $G$. Since every vertex in $R$ has at least two neighbors in $\tilde V_i$, $m_R \ge 2|R|$. 
Because $\tilde V_i \subseteq U_{c_i}$, Proposition~\ref{prop:balanced} guarantees $|\tilde V_i| \le \frac{1.1 n}{k}$ with high probability. Thus, by the $(\polylog n, \alpha)$-small set vertex expansion property (Definition~\ref{def:smallSetExpander}), the neighborhood size in $G$ satisfies $N_G(\tilde V_i) \ge d|\tilde V_i|(1-\alpha)$. 

The number of unique neighbors of $\tilde V_i$ in $G$ is therefore at least $d|\tilde V_i|(1-\alpha) - |R|$. Since the total number of outgoing edges from $\tilde V_i$ in $G$ is exactly $d|\tilde V_i|$, we have:
$$m_R + (d|\tilde V_i|(1-\alpha) - |R|) \le d|\tilde V_i| \implies m_R - |R| \le \alpha d |\tilde V_i|.$$
Using $m_R \ge 2|R|$, this bounds $m_R/2 \le \alpha d |\tilde V_i|$, and consequently $m_R \le 2\alpha d |\tilde V_i|$.
In $G_i$, the set of edges incident to $\tilde V_i$ is precisely $\tilde E_i$. The number of those edges connecting to $R$ is at most $m_R$. Thus, the number of unique neighbors in $G_i$ is   $|\tilde{S}_i| \ge |\tilde E_i| - m_R \ge |\tilde E_i| - 2\alpha d |\tilde{V}_i|$.
\end{proof}

Let $\tG_i$ be the bipartite subgraph of $G_i$ induced by the edges between $\tilde{V}_i$ and $\tilde{S}_i$. Let $\td_{i,v}$ be the degree of vertex $v$ in $\tG_i$. For all $v \in \tG_i$, $\td_{i,v} \le d_{i,v} \le d$. 
Since every vertex in $\tilde S_i$ has exactly one neighbor in $\tilde V_i$, $\tG_i$ is a collection of disjoint stars centered at $\tilde V_i$. Thus, the total number of edges in $\tG_i$ is exactly $|\tilde A_i| = |\tilde S_i|$. 

We establish the following relationship comparing the degree squares of $G_i$ and $\tilde{G}_i$:
\begin{claim}\label{claim:simple_focs}
    For $i \in I'$, $\frac{\sum_{v \in V_i} d_{i,v}^2}{|E_i|} \le \frac{\sum_{v \in \tilde{V}_i} \tilde{d}_{i,v}^2}{|\tilde{A}_i|} + 1 + 4d\left(\eta + \frac{3\alpha}{\tau}\right)$.
\end{claim}
\begin{proof}
For any $v \in V_i$, we have $0 \le \tilde{d}_{i,v} \le d_{i,v} \le d$. Expanding the square gives:
\[ \tilde{d}_{i,v}^2 = (d_{i,v} - (d_{i,v} - \tilde{d}_{i,v}))^2 \ge d_{i,v}^2 - 2d_{i,v}(d_{i,v} - \tilde{d}_{i,v}) \ge d_{i,v}^2 - 2d(d_{i,v} - \tilde{d}_{i,v}). \]
Summing over all $v \in V_i$, the sum of degrees in $G_i$ is $2|E_i|$ and in $\tilde{G}_i$ is $2|\tilde{A}_i|$. Thus $\sum (d_{i,v} - \tilde{d}_{i,v}) = 2(|E_i| - |\tilde{A}_i|)$.
Substituting this, we obtain $\sum \tilde{d}_{i,v}^2 \ge \sum d_{i,v}^2 - 4d(|E_i| - |\tilde{A}_i|)$.
Dividing by $|E_i|$ yields:
\[ \frac{\sum d_{i,v}^2}{|E_i|} \le \frac{\sum \tilde{d}_{i,v}^2}{|E_i|} + 4d\frac{|E_i| - |\tilde{A}_i|}{|E_i|}. \]
Because $|\tilde{A}_i| \le |E_i|$, replacing the denominator safely upper-bounds the fraction, yielding:
\[ \frac{\sum d_{i,v}^2}{|E_i|} \le \frac{\sum \tilde{d}_{i,v}^2}{|\tilde{A}_i|} + 4d\frac{|E_i| - |\tilde{A}_i|}{|E_i|}. \]
In the bipartite graph $\tilde{G}_i$, all edges are exclusively between $\tilde{V}_i$ and $\tilde{S}_i$. Since every vertex in $\tilde{S}_i$ has degree exactly $1$, we have $\sum_{v \in \tilde{S}_i} \tilde{d}_{i,v}^2 = |\tilde{S}_i| = |\tilde{A}_i|$. 
Thus, $\frac{\sum \tilde{d}_{i,v}^2}{|\tilde{A}_i|} = \frac{\sum_{\tilde{V}_i} \tilde{d}_{i,v}^2 + |\tilde{A}_i|}{|\tilde{A}_i|} = \frac{\sum_{\tilde{V}_i} \tilde{d}_{i,v}^2}{|\tilde{A}_i|} + 1$.

For the error term, Claim~\ref{claim:S_i} implies $|\tilde{A}_i| = |\tilde{S}_i| \ge |\tilde{E}_i| - 2\alpha d|\tilde{V}_i|$. Since $|\tilde{E}_i| \ge (1-\eta)|E_i|$, we have $|E_i| - |\tilde{A}_i| \le \eta |E_i| + 2\alpha d|\tilde{V}_i|$.
Because we explicitly restricted our evaluation to clusters $i \in I'$ which   satisfy $|E_i| \ge \tau \frac{nd}{k}$, and we bounded $|\tilde{V}_i| \le \frac{1.1n}{k}$, we cleanly bound the fraction:
\[ \frac{|E_i| - |\tilde{A}_i|}{|E_i|} \le \eta + \frac{2\alpha d(1.1n/k)}{\tau nd/k} = \eta + \frac{2.2\alpha}{\tau} \le \eta + \frac{3\alpha}{\tau}. \]
Substituting this completes the proof.
\end{proof}

From Lemma~\ref{lem:avgStruct}, since $I' \subseteq I$, for all $i \in I'$ we have $\frac{\sum_{V_i} d_{i,v}^2}{|E_i|} \ge d(1-\omega)$. Substituting our reduction from Claim~\ref{claim:simple_focs}, and noting that $|\tilde{A}_i| = \sum_{v \in \tilde{V}_i} \tilde{d}_{i,v}$, we isolate the ratio:
\[ \frac{\sum_{\tilde{V}_i} \tilde{d}_{i,v}^2}{\sum_{\tilde{V}_i} \tilde{d}_{i,v}} \ge d(1-\omega) - 1 - 4d\left(\eta + \frac{3\alpha}{\tau}\right) \ge d\left(1 - \omega - 4\eta - \frac{12\alpha}{\tau} - \frac{1}{d}\right). \]
Let $\theta := \omega + 4\eta + \frac{12\alpha}{\tau} + \frac{1}{d}$. We   lower bound this by $d(1-\theta)$. This directly implies:
\begin{equation}
    \sum_{v\in \tilde{V}_i}\td_{i,v} (\td_{i,v}-d(1-{\theta})) \ge 0. \label{eqregood}
\end{equation}

We  partition $\tV_i := \tV_i^+ \dot\cup L_i \dot\cup M_i \dot\cup R_i$ based on induced degree in $\tG_i$:
\begin{itemize}
    \item $\tV_i^+$: $\td_v \ge d(1-\theta)$.
    \item $L_i$: $d(1-2\sqrt{\theta}) \le \td_v < d(1-\theta)$.
    \item $M_i$: $2\sqrt{\theta}d < \td_v < d(1-2\sqrt{\theta})$.
    \item $R_i$: $\td_v \le 2\sqrt{\theta}d$.
\end{itemize}
Equation \eqref{eqregood} requires the positive contributions to outweigh the negative ones:
\begin{align}
\sum_{\tV_i^+}\td_v (\td_v-d(1-{\theta})) \ge \sum_{L_i \cup M_i \cup R_i}\td_v (d(1-{\theta})-\td_v).
\label{eqreregood}
\end{align}

Since $\td_v \le d$, the maximum positive contribution per vertex in $\tV_i^+$ is $d(d - d(1-\theta)) = \theta d^2$. Thus, the LHS is bounded by $|\tV_i^+| \theta d^2$.

For $v \in M_i$, the quadratic $\td_v(d(1-\theta) - \td_v)$ is a downward-facing parabola minimized at the endpoints of the interval $[2\sqrt{\theta}d, d(1-2\sqrt{\theta})]$. At either endpoint, the value is   at least $\frac{1}{2}\sqrt{\theta}d^2$. Thus, the penalty from $M_i$ is at least $|M_i| \frac{1}{2}\sqrt{\theta}d^2$. 
Comparing this to the LHS bound yields $|M_i| \frac{1}{2}\sqrt{\theta}d^2 \le |\tV_i^+| \theta d^2 \implies |M_i| \le 2\sqrt{\theta} |\tV_i^+|$.

We define $|\tilde A(S)| = \sum_{v \in S} \td_v$ as the number of edges in $\tG_i$ incident to a subset $S$. 
To bound $|\tilde A(M_i)|$, we choose $S_i \subseteq \tV_i^+$ to be the $|M_i|$ vertices in $\tV_i^+$ with the \emph{smallest} degrees. Because $S_i \subseteq \tV_i^+$, every $u \in S_i$ satisfies $\td_u \ge d(1-\theta) > \td_v$ for all $v \in M_i$. Bijecting $M_i$ to $S_i$, and noting that the average degree of the smallest elements is strictly less than the overall average, we obtain:
\[ |\tilde A(M_i)| = \sum_{v \in M_i} \td_v < \sum_{u \in S_i} \td_u \le |S_i| \frac{\sum_{\tV_i^+} \td_v}{|\tV_i^+|} = \frac{|M_i|}{|\tV_i^+|} |\tilde A(\tV_i^+)| \le 2\sqrt{\theta} |\tilde A(\tV_i^+)|. \]

For $R_i$, the penalty is $\td_v(d(1-\theta) - \td_v) \ge \td_v d(1-\theta - 2\sqrt{\theta}) \ge d(1-3\sqrt{\theta}) \td_v$. Summing over $R_i$ gives $d(1-3\sqrt{\theta}) |\tilde A(R_i)|$. 
The LHS of \eqref{eqreregood} is upper bounded by $d\theta \sum_{\tV_i^+} \td_v = \theta d |\tilde A(\tV_i^+)|$. Thus:
\[ d(1-3\sqrt{\theta}) |\tilde A(R_i)| \le \theta d |\tilde A(\tV_i^+)| \implies |\tilde A(R_i)| \le \frac{\theta}{1-3\sqrt{\theta}} |\tilde A(\tV_i^+)| \le 2\theta |\tilde A(\tV_i^+)|. \]

Summing the components, the total edges in $\tG_i$ is bounded by:
\begin{align*}
|\tilde A_i| &= |\tilde A(\tV_i^+)| + |\tilde A(L_i)| + |\tilde A(M_i)| + |\tilde A(R_i)| \\
&\le |\tilde A(\tV_i^+)| + |\tilde A(L_i)| + 2\sqrt{\theta} |\tilde A(\tV_i^+)| + 2\theta |\tilde A(\tV_i^+)| \\
&\le (1+4\sqrt{\theta}) \left( |\tilde A(\tV_i^+)| + |\tilde A(L_i)| \right).
\end{align*}

Define $U_i = \tV_i^+ \cup L_i$, and let $\tilde U = \bigcup_{i \in I'} U_i$. The number of edges in $\tG_i$ incident to $U_i$ is exactly $|\tilde A(U_i)| = |\tilde A(\tV_i^+)| + |\tilde A(L_i)|$. Thus, $|\tilde A_i| \le (1+4\sqrt{\theta}) |\tilde A(U_i)|$.

We relate the edge mass of $\tilde U$ directly to the global graph. Because $\tG_i$ is a disjoint star forest centered at $\tilde V_i$, and the edges $\tilde E_i$ represent disjoint subsets of the global edges $E$ (due to the distinct colors), there is absolutely no double-counting of edges across the sets $U_i$. Thus $\tilde U$ perfectly covers at least $\sum_{i \in I'} |\tilde A(U_i)|$ distinct edges in $G$. 

Summing $|\tilde A_i|$   over $I'$ and applying the  bound established in Claim~\ref{claim:simple_focs} (where $\frac{|E_i| - |\tilde{A}_i|}{|E_i|} \le \eta + \frac{3\alpha}{\tau} \implies |\tilde{A}_i| \ge |E_i| \left( 1 - \eta - \frac{3\alpha}{\tau} \right)$), we obtain:
\[ (1+4\sqrt{\theta}) |\tilde A(\tilde U)| \ge \sum_{i \in I'} |\tilde A_i| \ge \sum_{i \in I'} |E_i| \left(1 - \eta - \frac{3\alpha}{\tau}\right). \]
Using our safely retained edge mass bound $\sum_{i \in I'} |E_i| \ge (1-39\delta^2)|E|$ established at the beginning of this proof, we seamlessly bound the sum:
\[ \sum_{i \in I'} |\tilde A_i| \ge \left(1 - \eta - \frac{3\alpha}{\tau}\right) (1-39\delta^2)|E| \ge \left(1 - \eta - \frac{3\alpha}{\tau} - 39\delta^2\right)|E|. \]

Dividing by $(1+4\sqrt{\theta})$ and utilizing the algebraic identity $\frac{1-X}{1+Y} \ge 1 - X - Y$ (valid for $X,Y \ge 0$), we cleanly lower bound the covered edges without  inflating the error terms:
\[ |\tilde A(\tilde U)| \ge \frac{1 - \eta - 3\alpha/\tau - 39\delta^2}{1+4\sqrt{\theta}}|E| \ge \left(1 - \eta - \frac{3\alpha}{\tau} - 39\delta^2 - 4\sqrt{\theta}\right)|E|. \]

Finally, we upper bound the size of $\tilde U$. For every $i \in I'$ and $v \in U_i$, its induced degree in $\tG_i$ satisfies $\tilde d_{i,v} \ge d(1-2\sqrt{\theta})$. Since $\theta \approx \delta^3 \ll 1$, we safely have $1-2\sqrt{\theta} > 1/2$. Because the total degree of any vertex in the base graph $G$ is exactly $d$, it is impossible for a vertex to have strictly more than $d/2$ incident edges assigned to two different clusters. Thus, the sets $U_i$ are explicitly pairwise disjoint. 
This guarantees that $|\tilde U| = \sum_{i \in I'} |U_i|$ and that the total number of edges covered by $\tilde U$ across all $\tG_i$ evaluates exactly to $\sum_{v \in \tilde U} \tilde d_{i,v} = |\tilde A(\tilde U)|$.
Since the maximum number of edges in $G$ is $|E| = nd/2$, we have:
\[ d(1-2\sqrt{\theta}) |\tilde U| \le \sum_{v \in \tilde U} \tilde d_{i,v} = |\tilde A(\tilde U)| \le |E| = \frac{nd}{2}. \]
Rearranging this yields:
\[ |\tilde U| \le \frac{n}{2(1-2\sqrt{\theta})} \le \frac{n}{2}(1+4\sqrt{\theta}) = \frac{n}{2} + 2\sqrt{\theta}n. \]

Thus, we have found a subset $\tilde{U}$ of size at most $\frac{n}{2} + 2\sqrt{\theta}n$ that covers at least $(1 - \eta - 3\alpha/\tau - 39\delta^2 - 4\sqrt{\theta})|E|$ edges. 
To extract an exactly $n/2$-sized subset, we arbitrarily discard at most $2\sqrt{\theta}n$ vertices from $\tilde U$. Since the maximum degree in the base graph $G$ is $d$, discarding them sacrifices at most $2\sqrt{\theta}nd = 4\sqrt{\theta}|E|$ edges. If $|\tilde U| < n/2$, we pad it arbitrarily without losing edges.
The final remaining exactly $n/2$-sized subset   covers at least $(1 - \eta - 3\alpha/\tau - 39\delta^2 - 8\sqrt{\theta})|E|$.

As established earlier, $\tau = \delta^5$, $\eta = \delta^5$, $\alpha = \delta^{10}$, and $\omega = \delta^3$. Thus, $\alpha/\tau = \delta^5$, and $\theta = \delta^3 + 4\delta^5 + 12\delta^5 + 1/d$.
Thus, $\theta \le \delta^3 + 17\delta^5$. We bound the square root by $\sqrt{\theta} \le \delta^{1.5}\sqrt{1+17\delta^2} \le \delta^{1.5}(1+8.5\delta^2) = \delta^{1.5} + 8.5\delta^{3.5}$.
Then $8\sqrt{\theta} \le 8\delta^{1.5} + 68\delta^{3.5}$.

The total missed fraction is   bounded by:
\[ \text{Missed} \le \delta^5 + 3\delta^5 + 39\delta^2 + 8\delta^{1.5} + 68\delta^{3.5} \le 10\delta^{1.5}, \]
where the final inequality holds because Theorem~\ref{thm:lb_k_n} enforces $\delta \le 10^{-3}$.
\end{proof}

\section{Lower Bounds for Exact Algorithms: Proof of Corollary~\ref{cor:exact}}
 \label{sec:exact}
 In this section, we prove Corollary \ref{cor:exact}. Suppose that there is an algorithm solving exact $k$-means in $\R^d$ in times $T_{exact}(n,k,d)$.
The reduction proceeds as follows. Given a point set $P$, we compute a $\alpha$-coreset $\Omega$  consisting of $\text{poly}(k/\alpha)$ points. The definition of a $\alpha$-coreset ensures that, for any possible solution $S$, $\cost(P, S) = (1\pm \alpha)\cost(\Omega, S)$. In particular,  an optimal clustering on $\Omega$ is a $(1+\alpha)$-approximate clustering for $P$. The running time for that procedure is $T_{core}(n,d,k,\alpha^{-1}) = O(nkd)$ \cite{Cohen-AddadSS21}.

Next, we apply a Johnson-Lindenstrauss transformation on $\Omega$. Specifically, we use
Theorem 1.3 from \cite{MakarychevMR19}, which states that any set of $n$ points may be embedded into $O(\alpha^{-2}\log k/\alpha)$ dimensions such that for any partition into $k$ clusters, 
the $k$-means cost of the clustering is preserved up to a $(1\pm\alpha)$ factor. The running time for that dimension-reduction is $T_{JL}(n,d,k,\alpha^{-1}) = O(nd \log(k/\alpha))$.

Thus, by applying a coreset algorithm in time $T_{core}(n,d,k,\alpha^{-1})$, then a Johnson-Lindenstrauss transform in time $T_{JL}(|\Omega|,d,k,\alpha^{-1})$, and solving the resulting $k$-means instance exactly in time $T_{exact}(|\Omega|,k,\alpha^{-2}\log (k/\alpha))$, the overall running time for the algorithm is $T_{exact}(|\Omega|,k,\alpha^{-2}\log (k/\alpha))+T_{JL}(|\Omega|,d,k)+T_{core}(n,d,k)$.
We have $T_{core}(n,d,k,\alpha^{-1})+T_{JL}(|\Omega|,d,k,\alpha^{-1}) = O(nkd)$. 
Thus, from \cref{thm:lb_k}, it must hold $T_{exact}(|\Omega|,k,\alpha^{-2}\log k/\alpha)=2^{(k/\alpha)^{1-o(1)}}$.

Assume there exists an algorithm running in time $n^{(k\sqrt{d})^{1-\beta}}$. Running this algorithm on $\Omega$ would result in a running time of
\begin{align*}
 |\Omega|^{(k\sqrt{d})^{1-\beta}} &= 2^{(k\sqrt{d})^{1-\beta}\cdot \log |\Omega|} \\
 &= 2^{(k/\alpha)^{1-\beta}\cdot \polylog(k/\alpha)} \\
 &\le 2^{(k/\alpha)^{1-\beta/2}}.
\end{align*}

This contradicts \cref{thm:lb_k}, and therefore the running time of the exact algorithm must be $n^{(k\sqrt{d})^{1-o(1)}}$. This proves the first part of \Cref{cor:exact}.

To prove the second part, we use a different dimension-reduction: in time $\poly(n, d)$, one can reduce the dimension to $O(k/\alpha)$ \cite{cohen2015dimensionality}  while preserving the cost of any clustering. Therefore, the same argument as above enforces that $T_{exact}(|\Omega|, k, k/\alpha) \geq 2^{(\frac{k}{\alpha})^{1-o(1)}}$.

Suppose that $T_{exact}(n, k, d) = n^{d^{1-\beta}}$. Then,  
\begin{align*}
 T_{exact}(|\Omega|, k, k/\alpha) &= |\Omega|^{(k/\alpha)^{1-\beta}}\\
 &= 2^{(k/\alpha)^{1-\beta}\cdot \log |\Omega|}\\
 &= 2^{(k/\alpha)^{1-\beta}\cdot \polylog(k/\alpha))}\le 2^{(k/\alpha)^{1-\beta/2}},
\end{align*}
 which contradicts again \cref{thm:lb_k}. This concludes the proof of \Cref{cor:exact}.

\subsubsection*{Acknowledgements}
We thank Dor Minzer, Euiwoong Lee, and Pasin Manurangsi for several discussions that helped conceptualize the proof approach in this paper.

\bibliographystyle{alpha}
\bibliography{biblio}{}

\end{document}